\documentclass[showpacs,amsmath,amssymb,twocolumn,superscriptaddress,notitlepage,preprintnumbers,aps,pra,10pt,floatfix]{revtex4-2}
\usepackage{graphicx}
\usepackage{dcolumn}
\usepackage{bm}
\usepackage{qcircuit}
\usepackage{braket}
\usepackage{color}
\usepackage{amsthm}
\usepackage{chemarrow}
\usepackage{overpic}
\usepackage{diagbox}
\usepackage{braket}
\usepackage{subfigure}
\usepackage{makecell}
\newtheorem{theorem}{Theorem}
\newtheorem{definition}{Definition}
\newtheorem{lemma}{Lemma}
\newtheorem{corollary}{Corollary}
\usepackage{chngcntr}
\usepackage[colorlinks=true,linkcolor=red,bookmarks=true,breaklinks=true]{hyperref}
\definecolor{NG}{rgb}{0.66,0.29,0.48}

\begin{document}
\let\oldacl\addcontentsline
\renewcommand{\addcontentsline}[3]{}
\newcommand{\Tr}{\mathrm{Tr}}
\newcommand{\szx}[1]{\textcolor{blue}{(szx: #1)}}
\newcommand{\add}[1]{\textcolor{blue}{#1}}
\newcommand{\qi}[1]{\textcolor{red}{(qi: #1)}}
\title{Design nearly optimal quantum algorithm for linear differential equations via Lindbladians}
\author{Zhong-Xia Shang}
\email{shangzx@hku.hk}
\affiliation{HK Institute of Quantum Science $\&$ Technology, The University of Hong Kong, Hong Kong, China}
\affiliation{QICI Quantum Information and Computation Initiative, Department of Computer Science,
The University of Hong Kong, Hong Kong, China}
\author{Naixu Guo}
\email{naixug@u.nus.edu}
\affiliation{Centre for Quantum Technologies, National University of Singapore, 117543, Singapore}
\author{Dong An}
\email{dongan@pku.edu.cn}
\affiliation{Beijing International Center for Mathematical Research, Peking University, Beijing, China}
\author{Qi Zhao}
\email[]{zhaoqcs@hku.hk}
\affiliation{QICI Quantum Information and Computation Initiative, Department of Computer Science,
The University of Hong Kong, Hong Kong, China}
\begin{abstract}
Solving linear ordinary differential equations (ODE) is one of the most promising applications for quantum computers to demonstrate exponential advantages. The challenge of designing a quantum ODE algorithm is how to embed non-unitary dynamics into intrinsically unitary quantum circuits. In this work, we propose a new quantum algorithm for solving ODEs by harnessing open quantum systems. Specifically, we propose a novel technique called non-diagonal density matrix encoding, which leverages the inherent non-unitary dynamics of Lindbladians to encode general linear ODEs into the non-diagonal blocks of density matrices. This framework enables us to design quantum algorithms with both theoretical simplicity and high performance. Combined with the state-of-the-art quantum Lindbladian simulation algorithms, our algorithm can outperform all existing quantum ODE algorithms and achieve near-optimal dependence on all parameters under a plausible input model. We also give applications of our algorithm including the Gibbs state preparations and the partition function estimations.
\end{abstract}
\maketitle

\noindent\textit{\textbf{Introduction.—}}
Differential equations \cite{hartman2002ordinary} have long served as an essential tool for modeling and describing the dynamics of systems in both natural and social sciences. A general linear ordinary differential equation (ODE) can be written as
\begin{equation}\label{odedef}
\frac{d}{dt}\vec{\mu}(t) =-V(t)\vec{\mu}(t)+\vec{b}(t),
\end{equation}
where $V(t)\in \mathbb{C}^{2^n\times 2^n}$, $\vec{\mu}(t),\vec{b}(t)\in \mathbb{C}^{2^n}$, and the initial condition is given by $\vec{\mu}(0)=\vec{\mu}_0$. Classical simulation algorithms often become highly inefficient for large systems due to their polynomial scaling with respect to the system dimension. 
In contrast, quantum algorithms \cite{berry2014high,berry2017quantum,childs2020quantum,krovi2023improved,berry2024quantum,fang2023time,an2023linear,jin2022quantum,an2023quantum,low2024quantum} with appropriate input access can generate a quantum state encoding the solution of the ODE with only poly-logarithmic dependence on the system dimension. This remarkable efficiency makes solving ODE a promising application of quantum computers. 
In particular, the Hamiltonian simulation \cite{seth1996universal, Childs_2009, childs2012hamiltonians, berry2015simulating, low2017optimal, low2019hamiltonian}, which aims to simulate the Schrödinger equation, a special case of ODEs, on quantum computers, is arguably one of the most significant applications and may be among the first to demonstrate practical quantum advantages.

Since $V(t)$ is generally not anti-Hermitian, the time evolution operator for the ODE is not necessarily unitary. This is the case in many important problems, such as those arising in non-Hermitian physics~\cite{Bender2007,RegoMonteiroNobre2013,GiusteriMattiottiCelardo2015,yao2018edge,el2018non,GongAshidaKawabataEtAl2018,kawabata2019symmetry,okuma2020topological,AshidaGongUeda2020,MatsumotoKawabataAshidaEtAl2020,bergholtz2021exceptional,DingFangMa2022,ChenSongLado2023,ZhengQiaoWangEtAl2024,ShenLuLadoEtAl2024} and fluid dynamics~\cite{Batchelor_2000,evans2010partial,widder1976heat,cannon1984one,hundsdorfer2003numerical,thambynayagam2011diffusion}. The central challenge in designing quantum algorithms for linear ODEs is how to embed non-unitary dynamics into intrinsically unitary quantum dynamics. When $V(t)$ is a time-independent normal matrix and $\vec{b}(t) =\vec{0}$, the problem of solving an ODE can be efficiently addressed by the powerful quantum singular value transformation (QSVT) algorithm~\cite{gilyen2019quantum}. However, for ODEs with a time-dependent non-normal matrix $V(t)$ and a possibly inhomogeneous term $\vec{b}(t) \neq\vec{0}$, QSVT is no longer applicable due to the mismatch between the singular value transformation and the eigenvalue transformation.

Previous works have developed two main strategies for solving general linear ODEs. The first is the linear-system-based approach, which discretizes the ODE by a numerical scheme, reformulates the discretized ODE as a dilated linear system of equations, and solves the linear system by quantum linear system algorithms. 
The first efficient linear-system-based approach was proposed by Berry in~\cite{berry2014high}, applying multi-step methods for time discretization and the HHL algorithm~\cite{HarrowHassidimLloyd2009} for solving the resulting linear system. 
Since then, several subsequent works~\cite{berry2017quantum,childs2020quantum,krovi2023improved,berry2024quantum} have advanced the linear-system-based approach by employing higher-order time discretization and more advanced quantum linear system algorithms~\cite{Ambainis2012,ChildsKothariSomma2017,SubasiSommaOrsucci2019,LinTong2020,AnLin2022,CostaAnYuvalEtAl2022,Dalzell2024}. The second strategy is the evolution-based approach, which directly embeds the time evolution operator into the subspace of an efficiently implementable unitary by time-marching~\cite{fang2023time}, reducing to Hamiltonian simulation problems~\cite{jin2022quantum,an2023linear,an2023quantum}, or quantum eigenvalue processing~\cite{low2024quantum}.

Beyond these methodologies, it remains an intriguing challenge to develop quantum ODE algorithms with improved efficiency by reducing the ODE to other quantumly solvable tasks. For instance, it is rather tempting to use natural non-unitary dynamics of open quantum systems to solve non-unitary ODEs, such as the Lindblad master equation (Lindbladian) \cite{lindblad1976generators}. Unlike previous evolution-based algorithms where non-unitary dynamics is embedded into the \textbf{subspace} of unitaries (block encoding \cite{low2019hamiltonian}), dynamics of \textbf{subsystems} described by Lindbladians are naturally non-unitary due to the interaction with the environment. Lindbladian itself has been proven to be a universal quantum computing model \cite{verstraete2009quantum,chen2024local}. Recently, there have been several works proposing quantum algorithms based on Lindbladians, where they aim to encode states \cite{guo2024designing} of interest such as Gibbs states \cite{chen2023quantum,chen2023efficient,ding2024efficient,gilyen2024quantum} and ground states \cite{ding2024single,shang2024polynomial} as the steady states of Lindbladians. Inspired by these rapid developments, our question arises: 
\begin{center}
\textit{Is it possible to solve linear ODEs via Lindbladians?}
\end{center}
Since there already exist various quantum algorithms for simulating Lindbladians~\cite{kliesch2011dissipative,cleve2016efficient,childs2017efficient, schlimgen2021quantum,li2023opensystem,pocrnic2023quantum,he2024efficientoptimalcontrolopen, ding2024simulating}, which makes it rather amenable to solve linear ODEs via Lindbladians on quantum computers as long as we can establish a connection between them. 

In this work, we provide a positive answer to the above question. 
We demonstrate how general linear ODEs can be embedded into Lindbladians with the aid of a new technique called non-diagonal density matrix encoding. 
Based on this result, we develop an efficient quantum algorithm for solving ODEs by leveraging state-of-the-art Lindbladian simulation quantum algorithms \cite{li2023opensystem,he2024efficientoptimalcontrolopen}, considering two different input models for the coefficient matrix $V(t)$. Notably, under the second input model, our algorithm outperforms all existing quantum algorithms and achieves near-optimal dependence on all parameters. We would like to note that Ref.~\cite{kasatkin2023differential} also explores connections between ODEs and Lindbladians, using different techniques and mappings to show that any ODE can be associated with a Markovian quantum master equation, and providing conditions under which the master equation is Lindbladian. Therefore, the novelty of our work lies mainly in achieving a computational breakthrough, rather than in the observation of representability.

\begin{figure*}[htpb]
\includegraphics[width=0.65\linewidth]{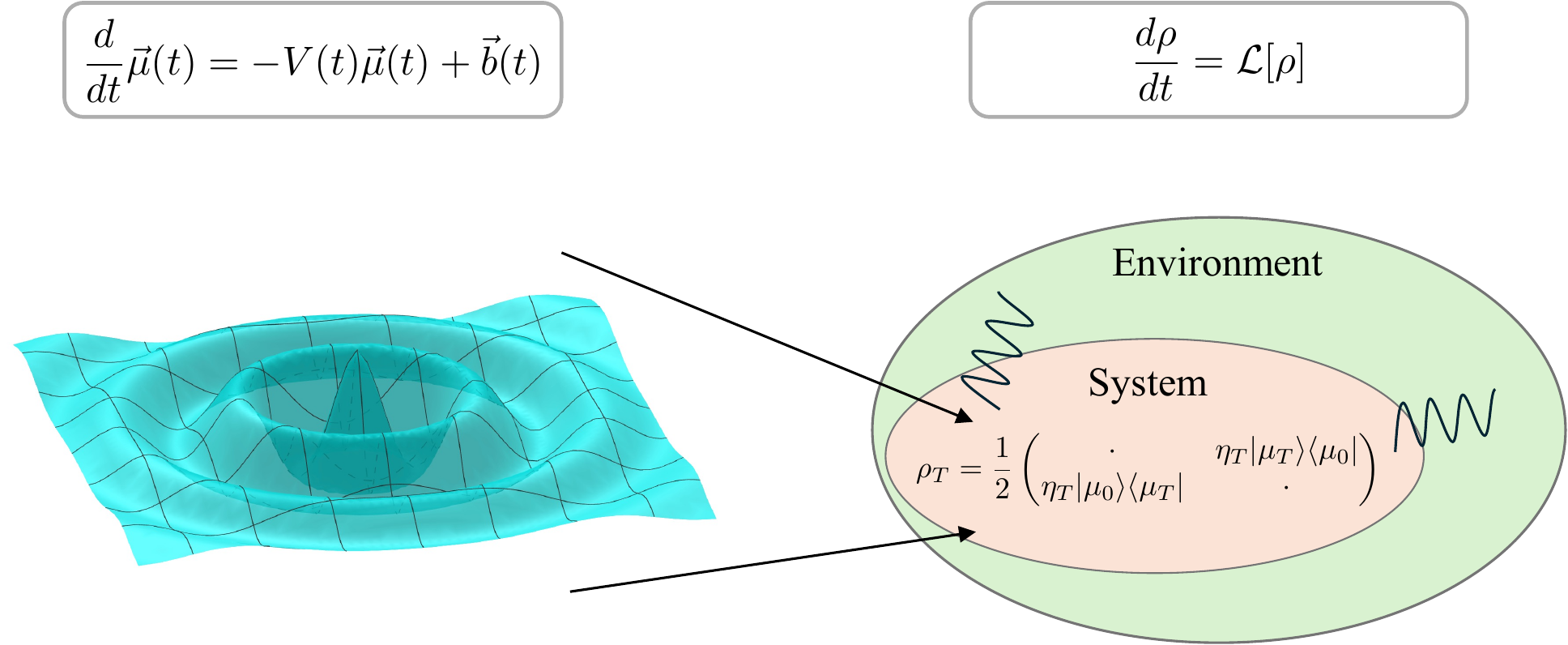}
\caption{High-level illustration of the algorithm. The dynamics of a linear ODE is embedded into the intrinsic non-unitary dynamics of a Lindbladian and the solution to the ODE is encoded into a non-diagonal block of the density matrix evolved according to the Lindbladian.
\label{fig1}}
\end{figure*}
\noindent\textit{\textbf{Algorithm overview.—}}
If a system is coupled to an environment that is sufficiently large for the Markovian approximation to hold, the dynamics of the system can be modeled by a Lindbladian \cite{lindblad1976generators}
\begin{eqnarray}\label{lmee}
&&\frac{d\rho}{dt}=\mathcal{L}[\rho]=-i[H(t),\rho]+\nonumber\\&&
\sum_i \left(F_i(t)\rho F_i(t)^\dag-
\frac{1}{2}\{\rho,F_i(t)^\dag F_i(t)\}\right),
\end{eqnarray}
where $\rho=\sum_{ij}\rho_{ij}|i\rangle\langle j|$ is the density matrix of the system, $H(t)$ is the internal Hamiltonian, and $\{F_i(t)\}$ are quantum jump operators. By treating $\rho$ as $\vec{u}(t)$ in ODE, the Lindbladian naturally corresponds to an ODE with a generally non-normal $V(t)$. However, the Lindbladian has to be a completely positive trace-preserving (CPTP) map \cite{nielsen2010quantum} and cannot be used to encode general ODE problems.

To address this issue, we propose a technique called non-diagonal density matrix encoding (NDME), firstly introduced in Ref \cite{shang2024estimating}. 
Similar to the idea of block encoding where a matrix $M$ of interest is encoded into the upper-right block of a unitary operator \cite{low2019hamiltonian} (also see the definition in SM \cite{supp} \ref{sec:qla}), NDME aims to encode the matrix $M$ into a non-diagonal block of a density matrix to jump out of the restrictions of Hermiticity and positive semi-definiteness:
\begin{definition}[Non-diagonal density matrix encoding (NDME)]
Given an $(l+n)$-qubit density matrix $\rho_M$ and an $n$-qubit matrix $M$, if $\rho_M$ satisfies
$(\langle s_1|_l\otimes I_n) \rho_M (|s_2\rangle_l\otimes I_n)=\gamma M$, where $|s_1\rangle_l$ and $|s_2\rangle_l$ are two different computational basis states of the $l$-qubit ancilla system and $\gamma\geq 0$, then $\rho_M$ is called an $(l+n,|s_1\rangle,|s_2\rangle,\gamma)$-NDME of $M$. 
\end{definition}  
\noindent In this work, we mainly focus on $(1+n,|0\rangle,|1\rangle,\gamma)$-NDME where $\rho_M$ has the form $
\rho_M=\begin{pmatrix}
\cdot & \gamma M \\
\gamma M^\dag & \cdot
\end{pmatrix}$. Since $\rho_M$ is Hermitian, it is also a $(1+n,|1\rangle,|0\rangle,\gamma)$-NDME of $M^\dag$.

We now show how to combine Lindbladians and NDME to encode general linear ODEs. Starting from a $(1+n)$-qubit initial state $\rho_0=|+\rangle\langle +|\otimes |\mu_0\rangle\langle\mu_0|$ which is a $(1+n,|0\rangle,|1\rangle,\frac{1}{2})$-NDME of $|\mu_0\rangle\langle \mu_0|$, we consider a Lindbladian of form Eq. \ref{lmee} with $H(t)=\begin{pmatrix}
H_1(t) & 0 \\
0 & \alpha I_n
\end{pmatrix}$ and $F_i(t)=\begin{pmatrix}
G_i(t) & 0 \\
0 & \beta_i I_n
\end{pmatrix}$. Focusing on the upper-right block of $\rho_0$, $\mathcal{L}[\rho_0]$ gives
\begin{eqnarray}
&&\frac{1}{2}|\mu_0\rangle\langle \mu_0|\rightarrow -i(H_1(t)-\alpha)\frac{1}{2}|\mu_0\rangle\langle \mu_0|-\\&&\left(\frac{1}{2}\sum_i G_i(t)^\dag G_i(t) -\sum_i\beta_i^*G_i(t)+\frac{1}{2}\sum_i|\beta_i|^2\right)\frac{1}{2}|\mu_0\rangle\langle \mu_0|.\nonumber
\end{eqnarray}
Here, since the $\alpha$ term can be absorbed into $H_1(t)$ and $\frac{1}{2}\sum_i G_i(t)^\dag G_i(t)$ is positive semi-definite, for semi-dissipative ODEs (the Hermitian part of $V(t)$ is positive semi-definite), we can simply set $\alpha=\beta_i=0$ such that on the left vector of this upper-right block, we realize the homogeneous linear ODE
\begin{eqnarray}\label{ooo}
\frac{d \vec{\mu}(t)}{dt}=-V(t)\vec{\mu}(t)=\left(-i H_1(t)-\frac{1}{2}\sum_i G_i(t)^\dag G_i(t)\right) \vec{\mu}(t),\nonumber\\
\end{eqnarray}
with $\vec{\mu}(0)=|\mu_0\rangle$. Eq.~\eqref{ooo} shows the interesting but expected connection between the anti-Hermitian (oscillating) part of $V(t)$ and the internal Hamiltonian of the Lindbladian, and the connection between the Hermitian (dissipative) part of $V(t)$ and the environment-induced jump operators of the Lindbladian. Evolving the Lindbladian for a time $T$, we will have
\begin{eqnarray}\label{rhott}
\rho_{T}=\frac{1}{2}\begin{pmatrix}
\cdot & \eta_T|\mu_T\rangle\langle \mu_0| \\
\eta_T|\mu_0\rangle\langle \mu_T| & \cdot
\end{pmatrix},
\end{eqnarray}
where $|\mu_T\rangle=\vec{\mu}(T)/\eta_T$ is the normalized solution with $\eta_T=\|\vec{\mu}(T)\|$. A high-level picture of these procedures is shown in Fig. \ref{fig1}. More details can be found in SM \cite{supp} \ref{sec: slovl}.

In the above discussions, we noted that it is possible to set $\beta_i\neq0$ for non-dissipative ODEs. However, we will focus on the semi-dissipative case because it has been shown that there is a no-go theorem with exponential worst-case complexity for the non-dissipative case \cite{an2023theoryquantumdifferentialequation}, which states that a generic quantum algorithm for ODEs with growing components must exhibit worst-case complexity of 
$e^{\Omega(T)}$. Such exponential scaling is considered inefficient. Consequently, existing quantum algorithms for ODEs are typically designed for the semi-dissipative case, and our work follows this established convention.

Having prepared $\rho_{T}$, we can further extract the information of $\vec{\mu}(T)$ depending on our purposes. Since the information of $\vec{\mu}(T)$ has already been encoded into $\rho_T$, we can perform measurements on $\rho_T$ to acquire properties of $\vec{\mu}(T)$ such as $\eta_T\langle \mu_0|\mu_T\rangle$ known as the Loschmidt echo \cite{goussev2012loschmidt} and the expectation value $\eta_T^2\langle \mu_T|O|\mu_T\rangle$ with respect to some observable $O$. We give concrete measurement strategies in SM \cite{supp} \ref{sec:mp}. 
With the aid of amplitude estimation algorithms \cite{brassard2002quantum,aaronson2020quantum,grinko2021iterative}, we can use $\mathcal{O}(\eta_T^{-1}\varepsilon^{-1})$ queries to the preparation oracle of $|S_{T}\rangle$, which is the purification of $\rho_T$, to give a $\eta_T\varepsilon$ additive estimation of $\eta_T\langle \mu_0|\mu_T\rangle$, and we can use $\mathcal{O}(\eta_T^{-2}\varepsilon^{-1})$ queries to the preparation unitary of $|S_{T}\rangle$ to give a $\eta_T^2\varepsilon$ additive estimation of $\eta_T^2\langle \mu_T|O|\mu_T\rangle$. We want to emphasize that the purification state $|S_T\rangle$ can be directly obtained from the Lindbladian simulation algorithms \cite{li2023opensystem,he2024efficientoptimalcontrolopen} we adopted in this work. Even if we only have access to $\rho_T$, we can still measure these properties by Hadamard tests \cite{aharonov2006polynomial} with a standard quantum limit scaling ($\varepsilon^{-2}$).

If our purpose is to prepare $|\mu_T\rangle$ rather than merely measuring properties, which is also the setting in other quantum algorithms for ODEs, we can instead use the amplitude amplification algorithm \cite{Brassard_2002} to extract $|\mu_T\rangle$ from $|S_T\rangle$. The observation is that we have the relation: $(\langle 0|\otimes I_n\otimes \langle E_{T1}|_e)|S_T\rangle=\eta_T/\sqrt{2}|\mu_T\rangle$ where $|E_{T1}\rangle_e$ is a special environment state, which tells how to design the Grover operator \cite{grover1996fast} for the amplitude amplification. As a result, to prepare $|\mu_T\rangle$, we require $\mathcal{O}(\eta_T^{-1})$ queries on the Lindbladian simulation oracle and the initial state ($|\mu_0\rangle$) preparation oracle which matches the lower bound \cite{an2023theoryquantumdifferentialequation}. See SM \cite{supp} \ref{sec: extract} for detailed constructions.

\noindent\textit{\textbf{Complexity analysis.—}}
\begin{table*}[t]
\renewcommand{\arraystretch}{2}
\centering
\begin{tabular}{c|c|c|c}\hline\hline
\textbf{Algorithms} & \makecell{\textbf{Direct access model}\\\textbf{Queries to $V(t)$}} & \makecell{\textbf{Square root access model}\\\textbf{Queries to $H_1(t)$ and $\{G_i(t)\}$}}& \textbf{Queries to $|\mu_0\rangle$} \\\hline 
Spectral method \cite{childs2020quantum} & $\widetilde{\mathcal{O}}\left( \eta_T^{-1}\kappa_V \alpha_VT \text{~poly}\left(\log\left(\frac{1}{\epsilon}\right)\right)\right)$  & no change&$\widetilde{\mathcal{O}}\left(\eta_T^{-1} \kappa_V  \alpha_VT \text{~poly}\left(\log\left(\frac{1}{\epsilon}\right)\right)\right)$ \\\hline
Truncated Dyson \cite{berry2024quantum} & $\widetilde{\mathcal{O}}\left( \eta_T^{-1} \alpha_VT \left(\log\left(\frac{1}{\epsilon}\right)\right)^2\right)$ &no change & $\mathcal{O}\left(\eta_T^{-1}  \alpha_VT \log\left(\frac{1}{\epsilon}\right) \right)$ \\\hline
\makecell{QEVT \cite{low2024quantum}(only\\ time-independent)} & $\widetilde{\mathcal{O}}\left(\eta_T^{-1} \left(\alpha_VT+ \log\left(\frac{1}{\epsilon}\right)\right)\log\left(\frac{1}{\epsilon}\right)\right) $  \cite{commentqevt}& no change&$\widetilde{\mathcal{O}}\left(\eta_T^{-1} \left( \alpha_VT+ \log\left(\frac{1}{\epsilon}\right)\right)\log\left(\frac{1}{\epsilon}\right)\right) $  \\\hline
Time-marching \cite{fang2023time} & $\widetilde{\mathcal{O}}\left( \eta_T^{-1} \alpha_V^2T^2 \log\left(\frac{1}{\epsilon}\right)\right)$ & no change& $\mathcal{O}\left(\eta_T^{-1} \right)$  \\\hline
LCHS \cite{an2023linear} & $\widetilde{\mathcal{O}}\left( \eta_T^{-2}  \alpha_VT \epsilon^{-1} \right)$  & no change & $\mathcal{O}\left(\eta_T^{-1} \right)$  \\\hline
\makecell{Improved LCHS \cite{an2023quantum}\\(time-independent)} & $\widetilde{\mathcal{O}}\left( \eta_T^{-1}  \alpha_VT \left(\log\left(\frac{1}{\epsilon}\right)\right)^{1+o(1)} \right)$  &no change  & $\mathcal{O}\left(\eta_T^{-1} \right)$  \\\hline
\makecell{Improved LCHS \cite{an2023quantum}\\(time-dependent)} & $\widetilde{\mathcal{O}}\left( \eta_T^{-1}  \alpha_VT \left(\log\left(\frac{1}{\epsilon}\right)\right)^{2+o(1)} \right)$ & no change & $\mathcal{O}\left(\eta_T^{-1} \right)$  \\\hline
\makecell{This work\\ (time-independent)}& $\widetilde{\mathcal{O}}\left(\eta_T^{-1}\Delta^{-1} \alpha_V^2 T\frac{\log^2(1/\epsilon)}{\log\log( 1/\epsilon)}\right)$ &$\widetilde{\mathcal{O}}\left(\eta_T^{-1}\alpha_VT\frac{\log(1/\epsilon)}{\log\log( 1/\epsilon)}\right)$ & $\mathcal{O}\left( \eta_T^{-1} \right)$  \\\hline
\makecell{This work \\(time-dependent)}& $\widetilde{\mathcal{O}}\left(\eta_T^{-1}\Delta^{-1}\alpha_V^2T\frac{\log^3(1/\epsilon)}{\log^2\log( 1/\epsilon)}\right)$&$\widetilde{\mathcal{O}}\left(\eta_T^{-1}\alpha_V T\frac{\log^2(1/\epsilon)}{\log^2\log( 1/\epsilon)}\right)$  & $\mathcal{O}\left( \eta_T^{-1} \right)$  \\\hline\hline
\end{tabular}
\caption{Comparison among various ODE solvers for homogeneous cases under two different input models,
direct access model (Theorem \ref{the1}) and square root access model (Theorem \ref{the2}). The symbol $\epsilon$ is defined as the preparation error of the normalized state $|\mu_T\rangle$. We define $\eta_T=\|\vec{\mu}(T)\|$ with $\vec{\mu}(0)=|\mu_0\rangle$. $\Delta$ is defined as the lower bound of the smallest non-zero eigenvalues of the Hermitian part of $V(t)$. $\kappa_V$ is the upper bound of condition numbers of $V(t)$.\label{tb1}}
\end{table*}
In this section, we discuss the overall complexity of our algorithm. To facilitate the comparison with previous algorithms, we focus on the preparation of $|\mu_T\rangle$ as the goal. Given an ODE in Eq. \ref{odedef} with $A(t)=\frac{V(t)-V^\dagger(t)}{2i}=H_1(t)$ and $B(t)=\frac{V(t)+V^\dagger(t)}{2}=\frac{1}{2}\sum_i G_i(t)^\dag G_i(t)$, the complexity now depends on two parts: input models and the Lindbladian simulation algorithms. For the Lindbladian simulation, we adopt the results in Ref. \cite{li2023opensystem} for the time-independent case, and the results in Ref. \cite{he2024efficientoptimalcontrolopen} for the time-dependent case. To enable the Lindbladian simulation, we need access (block encodings) of $H_1(t)$ and $\{G_i(t)\}$, which naturally leads to two different input models. 

The first model is that we are given the direct access (block encoding) of $V(t)$. In this case, we need to construct the block encodings of $H_1(t)$ and $\{G_i(t)\}$ from $V(t)$. For $\{G_i(t)\}$, we consider setting a single jump operator: $G(t)=\sqrt{2B(t)}$ constructed by implementing a modified square root function approximated by QSVT \cite{gilyen2019quantum,gilyén2022improvedquantumalgorithmsfidelity} using only odd polynomials (see SM \cite{supp} \ref{sec: algcom}), which introduces some overheads in terms of $\alpha_V$ and $\epsilon$ and also gives an additional complexity dependence on $\Delta$, the lower bound of the smallest non-zero eigenvalues of $B(t)$ \cite{commentdelta}. The second input model is the square root access model where we are given the block encoding of $H_1(t)$ and $\{G_i(t)\}$ in the form $\alpha_{H_1}^{-1}|0\rangle\langle 0|\otimes H_1(t)+\sum_i\alpha_{G_i}^{-1}|i\rangle\langle i|\otimes G_i(t)$. This is the most natural model for Lindbladian simulations and directly applies to several practical scenarios where the operator can be expressed in a sum-of-squares form. For example, consider the case where $B=\sum_{ij}h_{ij}c^\dag_i c_j$ is a free fermionic Hamiltonian with the matrix $h=\sum_{ij}h_{ij}|i\rangle\langle j|$.  In such cases, it is always possible to find a basis transformation $T$ such that $h=T^\dag\Sigma T$ with a diagonal $\Sigma$. Expressing $B$ in terms of the transformed creation and annihilation operators $c'^\dag_i$ and $c'_i$ with $c'_i=\sum_{j}T_{ij}c_j$ satisfies the square root access. For more general cases, we recommend readers to Ref. \cite{ashida2020nonhermitian,low2025fast,king2025quantumsimulationsumofsquaresspectral}.

Depending on the input models, we have the following two theorems: 
\begin{theorem}[Quantum ODE solver via Lindbladian, direct access model\label{the1}]
For semi-dissipative and homogeneous linear ODEs, assume we are given access to $\alpha_V$-block encoding $U_{V(t)}$ of $V(t)$ with $\alpha_V\geq \max_t \|V(t)\|$ and quantum state preparation unitary $U_{\mu_0}$: $\ket{0}_n\rightarrow |\mu_0\rangle$, and suppose that the smallest \textbf{nonzero} eigenvalues of $B(t)$ is lower bounded by $\Delta>0$. Then, there exists a quantum algorithm that outputs an $\epsilon$-approximation of $|\mu_T\rangle$ by using $\mathcal{O}(\eta_T^{-1})$ queries to $U_{\mu_0}$ (and its reverse) and $\widetilde{\mathcal{O}}\left(\eta_T^{-1}\Delta^{-1} \alpha_V^2 T\frac{\log^2(1/\epsilon)}{\log\log(1/\epsilon)}\right)$ queries to $U_{V(t)}$ (and its inverse) in the time-independent case, or $\widetilde{\mathcal{O}}\left(\eta_T^{-1}\Delta^{-1} \alpha_V^2 T\frac{\log^3(1/\epsilon)}{\log^2\log(1/\epsilon)}\right)$ queries to $U_{V(t)}$ (and its inverse) in the time-dependent case. 
\end{theorem}
\begin{theorem}[Quantum ODE solver via Lindbladian, square root access model\label{the2}]
For semi-dissipative and homogeneous linear ODEs, assume we are given access to $1$-block encoding $U_{H_1(t),G_i(t)}$ of $\alpha_{V}^{-1}|0\rangle\langle 0|\otimes H_1(t)+\sum_i\alpha_{V}^{-1/2}|i\rangle\langle i|\otimes G_i(t)$ with $H_1(t)=A(t)$, $\sum_i G_i(t)^\dag G_i(t)=B(t)$, and $\alpha_V\geq \max_t \|V(t)\|$, and access to the quantum state preparation unitary $U_{\mu_0}$: $\ket{0}_n\rightarrow |\mu_0\rangle$. Then, there exists a quantum algorithm that outputs an $\epsilon$-approximation of $|\mu_T\rangle$ by using $\mathcal{O}(\eta_T^{-1})$ queries to $U_{\mu_0}$ (and its reverse) and $\widetilde{\mathcal{O}}\left(\eta_T^{-1}\alpha_V T\frac{\log(1/\epsilon)}{\log\log(1/\epsilon)}\right)$ queries to $U_{H_1(t),G_i(t)}$ (and its inverse) in the time-independent case, or $\widetilde{\mathcal{O}}\left(\eta_T^{-1}\alpha_V T\frac{\log^2(1/\epsilon)}{\log^2\log(1/\epsilon)}\right)$ queries to $U_{H_1(t),G_i(t)}$ (and its inverse) in the time-dependent case.
\end{theorem}
\noindent In the descriptions of Theorem \ref{the1} and \ref{the2}, the block encoding model for time-dependent cases follows the HAM-T structure in Ref. \cite{low2018hamiltonian} with an additional time register. (See more details in SM \cite{supp} \ref{sec: algcom} where we also explicitly provide additional algorithmic costs, including gate counts and the number of required ancilla qubits.)

We summarize the complexity of our algorithm with comparisons with previous ones in Table \ref{tb1}. The lower bounds of solving ODEs with respect to $T$ is $\mathcal{O}(T)$ from the no fast-forwarding theorem \cite{berry2007efficient}, with respect to $\eta_T$ is $\mathcal{O}(\eta_T^{-1})$ originated from the state discrimination \cite{an2023theoryquantumdifferentialequation}, and with respect to $\epsilon$ is $\mathcal{O}(\frac{\log(1/\epsilon)}{\log\log(1/\epsilon)})$ from the parity checking \cite{berry2014exponential}. Under the direct access model, our algorithm gives near optimal dependence on $T$, $\eta_T$, and $\epsilon$. Under the square root access model, for previous methods, we need to reconstruct the block encoding of $V(t)$ from $U_{H_1(t),G_i(t)}$. Our algorithm outperforms all previous methods under this model. Compared with methods requiring quantum linear system solvers \cite{childs2020quantum,berry2024quantum,low2024quantum}, our algorithm achieves optimal queries to the initial state preparation. Compared with the time-marching method \cite{fang2023time}, our algorithm improves the dependence on $\alpha_VT$ from quadratic to linear. Compared with the state-of-the-art improved LCHS method \cite{an2023quantum}, our algorithm also gives better dependence on the preparation error $\epsilon$. Especially for the time-independent case, our algorithm achieves optimal dependence on $\epsilon$. We want to mention that while the two input models have different access, the complexity lower bound dependence on parameters is actually the same. (We give an explanation on this in SM \cite{supp} \ref{sec:lowb}.)

\noindent\textit{\textbf{Additional discussions.—\label{sec: addre}}}
\textit{Inhomogeneous case: }
Here, we give a high-level discussion on how to generalize our algorithm to inhomogeneous ODEs, which shares the similar strategy to Ref. \cite{an2023quantum}. By Duhamel’s principle \cite{hartman2002ordinary}, the solution of a general linear ODE in Eq. \ref{odedef} can be represented as 
\begin{equation}
\vec{\mu}(T)=\mathcal{T}e^{-\int_0^T V(t)dt}\vec{\mu}(0)+\int_0^T \mathcal{T}e^{-\int_t^T V(t')dt'}\vec{b}(t)dt,
\end{equation}
where $\mathcal{T}$ denotes the time-ordering operator. $\mathcal{T}e^{-\int_0^T V(t)dt}\vec{\mu}(0)$ is the solution for the homogeneous case and encoded into $\rho_{T}$. Following the same procedure as preparing $\rho_T$, each integrand $e^{-\int_t^T V(t')dt'}\vec{b}(t)$ at different time $t$ can be encoded into a density matrix $\rho_{b,t}$. We let $|S_{b,t}\rangle$ to be the purification of $\rho_{b,t}$. By discretizing the second term with $m$ slices, $\vec{\mu}(T)$ can be encoded into a density matrix $\rho_{sol}$ that roughly looks like $\rho_{sol}\propto \rho_T+T/m\sum_{t}\|\vec{b}(t)\|_2\rho_{b,t} $. One purification of $\rho_{sol}$ has the form: $|S_{sol}\rangle\propto |0\rangle|S_T\rangle+\sqrt{T/m}\sum_{t}\|\vec{b}(t)\|_2^{1/2}|t\rangle|S_{b,t}\rangle $, which can then be prepared by Lindbladian simulations with the time register controlling the evolution time on an initial state $|S_{0,sol}\rangle\propto |0\rangle|\mu_0\rangle+\sqrt{T/m}\sum_{t}\|\vec{b}(t)\|^{1/2}|t\rangle|b_t\rangle$.

\textit{Applications: } 
Differential equations can have vast applications in both academic and real life \cite{braun1983differential,hartman2002ordinary}. One particular topic we want to emphasize is that our algorithm also could have important applications in non-Hermitian physics~\cite{Bender2007,RegoMonteiroNobre2013,GiusteriMattiottiCelardo2015,yao2018edge,el2018non,GongAshidaKawabataEtAl2018,kawabata2019symmetry,okuma2020topological,AshidaGongUeda2020,MatsumotoKawabataAshidaEtAl2020,bergholtz2021exceptional,DingFangMa2022,ChenSongLado2023,ZhengQiaoWangEtAl2024,ShenLuLadoEtAl2024}. Since the dynamics governed by non-Hermitian Hamiltonians can be directly understood as general linear ODEs, our algorithm may set the playground of using quantum computers to test various non-Hermitian effects. While there exists the effective non-Hermitian Hamiltonian formalism \cite{roccati2022non}, it only approximates the Lindbladians well for short-time dynamics. In comparison, what we achieved here is to exactly embed effective non-Hermitian Hamiltonian dynamics through NDME and enable long-time behavior testing. (See SM \cite{supp} \ref{sec:rela} for detailed discussions.)

Our algorithm can also be directly applied to Gibbs state related tasks. Given an $n$-qubit positive and semi-definite Hamiltonian $B$, its corresponding Gibbs state at the temperature $1/\beta$ is $\rho_\beta=e^{-\beta B}/Z_\beta$ with $Z_\beta=\text{Tr}(e^{-\beta B})$ denoted as the partition function. Using our algorithm, the symmetric purification state $|\rho_\beta\rangle$ of $\rho_\beta$ can be prepared with the Grover-type speedup by noticing $|\rho_\beta\rangle=\sqrt{\frac{2^n}{Z_\beta}}\left(e^{-\beta/2 B}\otimes I_n\right)|\Omega\rangle$ with $|\Omega\rangle=2^{-n/2}(\sum_i |i\rangle|i\rangle)$ and $e^{-\beta/2 B}$, as the imaginary time evolution, is a special case of linear ODE. Also, since the value of the partition function $Z_\beta$ is directly encoded in $\rho_T$ (and $|S_T\rangle$, see SM \cite{supp} \ref{sec: extract}), with the aid of amplitude estimation algorithms \cite{Brassard_2002,aaronson2020quantum,grinko2021iterative}, our construction also gives way for efficient partition function estimations. Note that the ability for partition function estimation is not a general feature in other quantum ODE solvers, and only holds for some of the evolution-based methods~\cite{jin2022quantum,an2023linear,an2023quantum,low2024quantum}.

\noindent\textit{\textbf{Summary and outlook.—}}
In summary, we propose a new quantum algorithm for linear ODEs. By combining the natural non-unitary dynamics of Lindbladians and a new technique called the non-diagonal density matrix encoding, the solution to a differential equation can be encoded into a non-diagonal block of a density matrix. With the aid of advanced Lindbladian simulation algorithms and the assumption of a plausible input model, we achieve near-optimal dependence on all parameters and outperform all existing methods. We also talk about its potential influence on non-Hermitian physics and Gibbs-related applications.

For future directions, since the complexity heavily depends on the results of Lindbladian simulations, our results give motivations to further improve the results for Lindbladian simulations, especially for the time-dependent case. In this sense, it might be easier to improve our algorithm to a lower complexity compared with the improved LCHS method \cite{an2023quantum}. Beyond solving differential equations, Since our algorithm achieves eigenvalue transformations beyond the framework of QSVT \cite{gilyen2019quantum}, it will also be interesting to explore whether the idea in this work can be used for other eigenvalue-related problems such as solving general systems of linear equations. We hope this work can invoke future studies on Lindbladian-based quantum algorithms and other quantum applications based on non-diagonal density matrix encoding.





\noindent\textit{\textbf{Achknowledgements.—}}
The authors would like to thank Zongping Gong, Chao-Yang Lu, Andris Ambainis, Wenjun Yu, Tianfeng Feng, Yuchen Guo, Shi Jin, Patrick Rebentrost, Chunhao Wang, and Xiantao Li for fruitful discussions. Z.S. and Q.Z. acknowledge the support from the HKU Seed Fund for Basic
Research for New Staff via Project 2201100596, Guangdong Natural Science Fund via Project 2023A1515012185, National Natural Science Foundation of China (NSFC) via Project No. 12305030 and No. 12347104, Hong Kong Research Grant Council (RGC) via No. 27300823, N\_HKU718/23, and R6010-23, Guangdong Provincial Quantum Science Strategic Initiative GDZX2200001. 
\bibliography{ref}

\begin{thebibliography}{101}%
\makeatletter
\providecommand \@ifxundefined [1]{%
 \@ifx{#1\undefined}
}%
\providecommand \@ifnum [1]{%
 \ifnum #1\expandafter \@firstoftwo
 \else \expandafter \@secondoftwo
 \fi
}%
\providecommand \@ifx [1]{%
 \ifx #1\expandafter \@firstoftwo
 \else \expandafter \@secondoftwo
 \fi
}%
\providecommand \natexlab [1]{#1}%
\providecommand \enquote  [1]{``#1''}%
\providecommand \bibnamefont  [1]{#1}%
\providecommand \bibfnamefont [1]{#1}%
\providecommand \citenamefont [1]{#1}%
\providecommand \href@noop [0]{\@secondoftwo}%
\providecommand \href [0]{\begingroup \@sanitize@url \@href}%
\providecommand \@href[1]{\@@startlink{#1}\@@href}%
\providecommand \@@href[1]{\endgroup#1\@@endlink}%
\providecommand \@sanitize@url [0]{\catcode `\\12\catcode `\$12\catcode `\&12\catcode `\#12\catcode `\^12\catcode `\_12\catcode `\%12\relax}%
\providecommand \@@startlink[1]{}%
\providecommand \@@endlink[0]{}%
\providecommand \url  [0]{\begingroup\@sanitize@url \@url }%
\providecommand \@url [1]{\endgroup\@href {#1}{\urlprefix }}%
\providecommand \urlprefix  [0]{URL }%
\providecommand \Eprint [0]{\href }%
\providecommand \doibase [0]{https://doi.org/}%
\providecommand \selectlanguage [0]{\@gobble}%
\providecommand \bibinfo  [0]{\@secondoftwo}%
\providecommand \bibfield  [0]{\@secondoftwo}%
\providecommand \translation [1]{[#1]}%
\providecommand \BibitemOpen [0]{}%
\providecommand \bibitemStop [0]{}%
\providecommand \bibitemNoStop [0]{.\EOS\space}%
\providecommand \EOS [0]{\spacefactor3000\relax}%
\providecommand \BibitemShut  [1]{\csname bibitem#1\endcsname}%
\let\auto@bib@innerbib\@empty
\bibitem [{\citenamefont {Hartman}(2002)}]{hartman2002ordinary}%
  \BibitemOpen
  \bibfield  {author} {\bibinfo {author} {\bibfnamefont {P.}~\bibnamefont {Hartman}},\ }\href@noop {} {\emph {\bibinfo {title} {Ordinary differential equations}}}\ (\bibinfo  {publisher} {SIAM},\ \bibinfo {year} {2002})\BibitemShut {NoStop}%
\bibitem [{\citenamefont {Berry}(2014)}]{berry2014high}%
  \BibitemOpen
  \bibfield  {author} {\bibinfo {author} {\bibfnamefont {D.~W.}\ \bibnamefont {Berry}},\ }\bibfield  {title} {\bibinfo {title} {High-order quantum algorithm for solving linear differential equations},\ }\href@noop {} {\bibfield  {journal} {\bibinfo  {journal} {Journal of Physics A: Mathematical and Theoretical}\ }\textbf {\bibinfo {volume} {47}},\ \bibinfo {pages} {105301} (\bibinfo {year} {2014})}\BibitemShut {NoStop}%
\bibitem [{\citenamefont {Berry}\ \emph {et~al.}(2017)\citenamefont {Berry}, \citenamefont {Childs}, \citenamefont {Ostrander},\ and\ \citenamefont {Wang}}]{berry2017quantum}%
  \BibitemOpen
  \bibfield  {author} {\bibinfo {author} {\bibfnamefont {D.~W.}\ \bibnamefont {Berry}}, \bibinfo {author} {\bibfnamefont {A.~M.}\ \bibnamefont {Childs}}, \bibinfo {author} {\bibfnamefont {A.}~\bibnamefont {Ostrander}},\ and\ \bibinfo {author} {\bibfnamefont {G.}~\bibnamefont {Wang}},\ }\bibfield  {title} {\bibinfo {title} {Quantum algorithm for linear differential equations with exponentially improved dependence on precision},\ }\href@noop {} {\bibfield  {journal} {\bibinfo  {journal} {Communications in Mathematical Physics}\ }\textbf {\bibinfo {volume} {356}},\ \bibinfo {pages} {1057} (\bibinfo {year} {2017})}\BibitemShut {NoStop}%
\bibitem [{\citenamefont {Childs}\ and\ \citenamefont {Liu}(2020)}]{childs2020quantum}%
  \BibitemOpen
  \bibfield  {author} {\bibinfo {author} {\bibfnamefont {A.~M.}\ \bibnamefont {Childs}}\ and\ \bibinfo {author} {\bibfnamefont {J.-P.}\ \bibnamefont {Liu}},\ }\bibfield  {title} {\bibinfo {title} {Quantum spectral methods for differential equations},\ }\href@noop {} {\bibfield  {journal} {\bibinfo  {journal} {Communications in Mathematical Physics}\ }\textbf {\bibinfo {volume} {375}},\ \bibinfo {pages} {1427} (\bibinfo {year} {2020})}\BibitemShut {NoStop}%
\bibitem [{\citenamefont {Krovi}(2023)}]{krovi2023improved}%
  \BibitemOpen
  \bibfield  {author} {\bibinfo {author} {\bibfnamefont {H.}~\bibnamefont {Krovi}},\ }\bibfield  {title} {\bibinfo {title} {Improved quantum algorithms for linear and nonlinear differential equations},\ }\href@noop {} {\bibfield  {journal} {\bibinfo  {journal} {Quantum}\ }\textbf {\bibinfo {volume} {7}},\ \bibinfo {pages} {913} (\bibinfo {year} {2023})}\BibitemShut {NoStop}%
\bibitem [{\citenamefont {Berry}\ and\ \citenamefont {Costa}(2024)}]{berry2024quantum}%
  \BibitemOpen
  \bibfield  {author} {\bibinfo {author} {\bibfnamefont {D.~W.}\ \bibnamefont {Berry}}\ and\ \bibinfo {author} {\bibfnamefont {P.~C.}\ \bibnamefont {Costa}},\ }\bibfield  {title} {\bibinfo {title} {Quantum algorithm for time-dependent differential equations using dyson series},\ }\href@noop {} {\bibfield  {journal} {\bibinfo  {journal} {Quantum}\ }\textbf {\bibinfo {volume} {8}},\ \bibinfo {pages} {1369} (\bibinfo {year} {2024})}\BibitemShut {NoStop}%
\bibitem [{\citenamefont {Fang}\ \emph {et~al.}(2023)\citenamefont {Fang}, \citenamefont {Lin},\ and\ \citenamefont {Tong}}]{fang2023time}%
  \BibitemOpen
  \bibfield  {author} {\bibinfo {author} {\bibfnamefont {D.}~\bibnamefont {Fang}}, \bibinfo {author} {\bibfnamefont {L.}~\bibnamefont {Lin}},\ and\ \bibinfo {author} {\bibfnamefont {Y.}~\bibnamefont {Tong}},\ }\bibfield  {title} {\bibinfo {title} {Time-marching based quantum solvers for time-dependent linear differential equations},\ }\href@noop {} {\bibfield  {journal} {\bibinfo  {journal} {Quantum}\ }\textbf {\bibinfo {volume} {7}},\ \bibinfo {pages} {955} (\bibinfo {year} {2023})}\BibitemShut {NoStop}%
\bibitem [{\citenamefont {An}\ \emph {et~al.}(2023{\natexlab{a}})\citenamefont {An}, \citenamefont {Liu},\ and\ \citenamefont {Lin}}]{an2023linear}%
  \BibitemOpen
  \bibfield  {author} {\bibinfo {author} {\bibfnamefont {D.}~\bibnamefont {An}}, \bibinfo {author} {\bibfnamefont {J.-P.}\ \bibnamefont {Liu}},\ and\ \bibinfo {author} {\bibfnamefont {L.}~\bibnamefont {Lin}},\ }\bibfield  {title} {\bibinfo {title} {Linear combination of hamiltonian simulation for nonunitary dynamics with optimal state preparation cost},\ }\href@noop {} {\bibfield  {journal} {\bibinfo  {journal} {Physical Review Letters}\ }\textbf {\bibinfo {volume} {131}},\ \bibinfo {pages} {150603} (\bibinfo {year} {2023}{\natexlab{a}})}\BibitemShut {NoStop}%
\bibitem [{\citenamefont {Jin}\ \emph {et~al.}(2022)\citenamefont {Jin}, \citenamefont {Liu},\ and\ \citenamefont {Yu}}]{jin2022quantum}%
  \BibitemOpen
  \bibfield  {author} {\bibinfo {author} {\bibfnamefont {S.}~\bibnamefont {Jin}}, \bibinfo {author} {\bibfnamefont {N.}~\bibnamefont {Liu}},\ and\ \bibinfo {author} {\bibfnamefont {Y.}~\bibnamefont {Yu}},\ }\bibfield  {title} {\bibinfo {title} {Quantum simulation of partial differential equations via schrodingerisation: technical details},\ }\href@noop {} {\bibfield  {journal} {\bibinfo  {journal} {arXiv preprint arXiv:2212.14703}\ } (\bibinfo {year} {2022})}\BibitemShut {NoStop}%
\bibitem [{\citenamefont {An}\ \emph {et~al.}(2023{\natexlab{b}})\citenamefont {An}, \citenamefont {Childs},\ and\ \citenamefont {Lin}}]{an2023quantum}%
  \BibitemOpen
  \bibfield  {author} {\bibinfo {author} {\bibfnamefont {D.}~\bibnamefont {An}}, \bibinfo {author} {\bibfnamefont {A.~M.}\ \bibnamefont {Childs}},\ and\ \bibinfo {author} {\bibfnamefont {L.}~\bibnamefont {Lin}},\ }\bibfield  {title} {\bibinfo {title} {Quantum algorithm for linear non-unitary dynamics with near-optimal dependence on all parameters},\ }\href@noop {} {\bibfield  {journal} {\bibinfo  {journal} {arXiv preprint arXiv:2312.03916}\ } (\bibinfo {year} {2023}{\natexlab{b}})}\BibitemShut {NoStop}%
\bibitem [{\citenamefont {Low}\ and\ \citenamefont {Su}(2024)}]{low2024quantum}%
  \BibitemOpen
  \bibfield  {author} {\bibinfo {author} {\bibfnamefont {G.~H.}\ \bibnamefont {Low}}\ and\ \bibinfo {author} {\bibfnamefont {Y.}~\bibnamefont {Su}},\ }\bibfield  {title} {\bibinfo {title} {Quantum eigenvalue processing},\ }\href@noop {} {\bibfield  {journal} {\bibinfo  {journal} {arXiv preprint arXiv:2401.06240}\ } (\bibinfo {year} {2024})}\BibitemShut {NoStop}%
\bibitem [{\citenamefont {Lloyd}(1996)}]{seth1996universal}%
  \BibitemOpen
  \bibfield  {author} {\bibinfo {author} {\bibfnamefont {S.}~\bibnamefont {Lloyd}},\ }\bibfield  {title} {\bibinfo {title} {Universal quantum simulators},\ }\href {https://doi.org/10.1126/science.273.5278.1073} {\bibfield  {journal} {\bibinfo  {journal} {Science}\ }\textbf {\bibinfo {volume} {273}},\ \bibinfo {pages} {1073} (\bibinfo {year} {1996})}\BibitemShut {NoStop}%
\bibitem [{\citenamefont {Childs}(2009)}]{Childs_2009}%
  \BibitemOpen
  \bibfield  {author} {\bibinfo {author} {\bibfnamefont {A.~M.}\ \bibnamefont {Childs}},\ }\bibfield  {title} {\bibinfo {title} {On the relationship between continuous- and discrete-time quantum walk},\ }\href {https://doi.org/10.1007/s00220-009-0930-1} {\bibfield  {journal} {\bibinfo  {journal} {Communications in Mathematical Physics}\ }\textbf {\bibinfo {volume} {294}},\ \bibinfo {pages} {581–603} (\bibinfo {year} {2009})}\BibitemShut {NoStop}%
\bibitem [{\citenamefont {Childs}\ and\ \citenamefont {Wiebe}(2012)}]{childs2012hamiltonians}%
  \BibitemOpen
  \bibfield  {author} {\bibinfo {author} {\bibfnamefont {A.~M.}\ \bibnamefont {Childs}}\ and\ \bibinfo {author} {\bibfnamefont {N.}~\bibnamefont {Wiebe}},\ }\bibfield  {title} {\bibinfo {title} {Hamiltonian simulation using linear combinations of unitary operations},\ }\href@noop {} {\bibfield  {journal} {\bibinfo  {journal} {arXiv preprint arXiv:1202.5822}\ } (\bibinfo {year} {2012})}\BibitemShut {NoStop}%
\bibitem [{\citenamefont {Berry}\ \emph {et~al.}(2015)\citenamefont {Berry}, \citenamefont {Childs}, \citenamefont {Cleve}, \citenamefont {Kothari},\ and\ \citenamefont {Somma}}]{berry2015simulating}%
  \BibitemOpen
  \bibfield  {author} {\bibinfo {author} {\bibfnamefont {D.~W.}\ \bibnamefont {Berry}}, \bibinfo {author} {\bibfnamefont {A.~M.}\ \bibnamefont {Childs}}, \bibinfo {author} {\bibfnamefont {R.}~\bibnamefont {Cleve}}, \bibinfo {author} {\bibfnamefont {R.}~\bibnamefont {Kothari}},\ and\ \bibinfo {author} {\bibfnamefont {R.~D.}\ \bibnamefont {Somma}},\ }\bibfield  {title} {\bibinfo {title} {Simulating hamiltonian dynamics with a truncated taylor series},\ }\href {https://doi.org/10.1103/PhysRevLett.114.090502} {\bibfield  {journal} {\bibinfo  {journal} {Phys. Rev. Lett.}\ }\textbf {\bibinfo {volume} {114}},\ \bibinfo {pages} {090502} (\bibinfo {year} {2015})}\BibitemShut {NoStop}%
\bibitem [{\citenamefont {Low}\ and\ \citenamefont {Chuang}(2017)}]{low2017optimal}%
  \BibitemOpen
  \bibfield  {author} {\bibinfo {author} {\bibfnamefont {G.~H.}\ \bibnamefont {Low}}\ and\ \bibinfo {author} {\bibfnamefont {I.~L.}\ \bibnamefont {Chuang}},\ }\bibfield  {title} {\bibinfo {title} {Optimal hamiltonian simulation by quantum signal processing},\ }\href {https://doi.org/10.1103/PhysRevLett.118.010501} {\bibfield  {journal} {\bibinfo  {journal} {Phys. Rev. Lett.}\ }\textbf {\bibinfo {volume} {118}},\ \bibinfo {pages} {010501} (\bibinfo {year} {2017})}\BibitemShut {NoStop}%
\bibitem [{\citenamefont {Low}\ and\ \citenamefont {Chuang}(2019)}]{low2019hamiltonian}%
  \BibitemOpen
  \bibfield  {author} {\bibinfo {author} {\bibfnamefont {G.~H.}\ \bibnamefont {Low}}\ and\ \bibinfo {author} {\bibfnamefont {I.~L.}\ \bibnamefont {Chuang}},\ }\bibfield  {title} {\bibinfo {title} {Hamiltonian {S}imulation by {Q}ubitization},\ }\href {https://doi.org/10.22331/q-2019-07-12-163} {\bibfield  {journal} {\bibinfo  {journal} {{Quantum}}\ }\textbf {\bibinfo {volume} {3}},\ \bibinfo {pages} {163} (\bibinfo {year} {2019})}\BibitemShut {NoStop}%
\bibitem [{\citenamefont {Bender}(2007)}]{Bender2007}%
  \BibitemOpen
  \bibfield  {author} {\bibinfo {author} {\bibfnamefont {C.~M.}\ \bibnamefont {Bender}},\ }\bibfield  {title} {\bibinfo {title} {Making sense of non-hermitian hamiltonians},\ }\href {https://doi.org/10.1088/0034-4885/70/6/r03} {\bibfield  {journal} {\bibinfo  {journal} {Reports on Progress in Physics}\ }\textbf {\bibinfo {volume} {70}},\ \bibinfo {pages} {947–1018} (\bibinfo {year} {2007})}\BibitemShut {NoStop}%
\bibitem [{\citenamefont {Rego-Monteiro}\ and\ \citenamefont {Nobre}(2013)}]{RegoMonteiroNobre2013}%
  \BibitemOpen
  \bibfield  {author} {\bibinfo {author} {\bibfnamefont {M.~A.}\ \bibnamefont {Rego-Monteiro}}\ and\ \bibinfo {author} {\bibfnamefont {F.~D.}\ \bibnamefont {Nobre}},\ }\bibfield  {title} {\bibinfo {title} {Classical field theory for a non-hermitian schr\"odinger equation with position-dependent masses},\ }\href {https://doi.org/10.1103/PhysRevA.88.032105} {\bibfield  {journal} {\bibinfo  {journal} {Phys. Rev. A}\ }\textbf {\bibinfo {volume} {88}},\ \bibinfo {pages} {032105} (\bibinfo {year} {2013})}\BibitemShut {NoStop}%
\bibitem [{\citenamefont {Giusteri}\ \emph {et~al.}(2015)\citenamefont {Giusteri}, \citenamefont {Mattiotti},\ and\ \citenamefont {Celardo}}]{GiusteriMattiottiCelardo2015}%
  \BibitemOpen
  \bibfield  {author} {\bibinfo {author} {\bibfnamefont {G.~G.}\ \bibnamefont {Giusteri}}, \bibinfo {author} {\bibfnamefont {F.}~\bibnamefont {Mattiotti}},\ and\ \bibinfo {author} {\bibfnamefont {G.~L.}\ \bibnamefont {Celardo}},\ }\bibfield  {title} {\bibinfo {title} {Non-hermitian hamiltonian approach to quantum transport in disordered networks with sinks: Validity and effectiveness},\ }\bibfield  {journal} {\bibinfo  {journal} {Physical Review B}\ }\textbf {\bibinfo {volume} {91}},\ \href {https://doi.org/10.1103/physrevb.91.094301} {10.1103/physrevb.91.094301} (\bibinfo {year} {2015})\BibitemShut {NoStop}%
\bibitem [{\citenamefont {Yao}\ and\ \citenamefont {Wang}(2018)}]{yao2018edge}%
  \BibitemOpen
  \bibfield  {author} {\bibinfo {author} {\bibfnamefont {S.}~\bibnamefont {Yao}}\ and\ \bibinfo {author} {\bibfnamefont {Z.}~\bibnamefont {Wang}},\ }\bibfield  {title} {\bibinfo {title} {Edge states and topological invariants of non-hermitian systems},\ }\href@noop {} {\bibfield  {journal} {\bibinfo  {journal} {Physical review letters}\ }\textbf {\bibinfo {volume} {121}},\ \bibinfo {pages} {086803} (\bibinfo {year} {2018})}\BibitemShut {NoStop}%
\bibitem [{\citenamefont {El-Ganainy}\ \emph {et~al.}(2018)\citenamefont {El-Ganainy}, \citenamefont {Makris}, \citenamefont {Khajavikhan}, \citenamefont {Musslimani}, \citenamefont {Rotter},\ and\ \citenamefont {Christodoulides}}]{el2018non}%
  \BibitemOpen
  \bibfield  {author} {\bibinfo {author} {\bibfnamefont {R.}~\bibnamefont {El-Ganainy}}, \bibinfo {author} {\bibfnamefont {K.~G.}\ \bibnamefont {Makris}}, \bibinfo {author} {\bibfnamefont {M.}~\bibnamefont {Khajavikhan}}, \bibinfo {author} {\bibfnamefont {Z.~H.}\ \bibnamefont {Musslimani}}, \bibinfo {author} {\bibfnamefont {S.}~\bibnamefont {Rotter}},\ and\ \bibinfo {author} {\bibfnamefont {D.~N.}\ \bibnamefont {Christodoulides}},\ }\bibfield  {title} {\bibinfo {title} {Non-hermitian physics and pt symmetry},\ }\href@noop {} {\bibfield  {journal} {\bibinfo  {journal} {Nature Physics}\ }\textbf {\bibinfo {volume} {14}},\ \bibinfo {pages} {11} (\bibinfo {year} {2018})}\BibitemShut {NoStop}%
\bibitem [{\citenamefont {Gong}\ \emph {et~al.}(2018)\citenamefont {Gong}, \citenamefont {Ashida}, \citenamefont {Kawabata}, \citenamefont {Takasan}, \citenamefont {Higashikawa},\ and\ \citenamefont {Ueda}}]{GongAshidaKawabataEtAl2018}%
  \BibitemOpen
  \bibfield  {author} {\bibinfo {author} {\bibfnamefont {Z.}~\bibnamefont {Gong}}, \bibinfo {author} {\bibfnamefont {Y.}~\bibnamefont {Ashida}}, \bibinfo {author} {\bibfnamefont {K.}~\bibnamefont {Kawabata}}, \bibinfo {author} {\bibfnamefont {K.}~\bibnamefont {Takasan}}, \bibinfo {author} {\bibfnamefont {S.}~\bibnamefont {Higashikawa}},\ and\ \bibinfo {author} {\bibfnamefont {M.}~\bibnamefont {Ueda}},\ }\bibfield  {title} {\bibinfo {title} {Topological phases of non-hermitian systems},\ }\bibfield  {journal} {\bibinfo  {journal} {Physical Review X}\ }\textbf {\bibinfo {volume} {8}},\ \href {https://doi.org/10.1103/physrevx.8.031079} {10.1103/physrevx.8.031079} (\bibinfo {year} {2018})\BibitemShut {NoStop}%
\bibitem [{\citenamefont {Kawabata}\ \emph {et~al.}(2019)\citenamefont {Kawabata}, \citenamefont {Shiozaki}, \citenamefont {Ueda},\ and\ \citenamefont {Sato}}]{kawabata2019symmetry}%
  \BibitemOpen
  \bibfield  {author} {\bibinfo {author} {\bibfnamefont {K.}~\bibnamefont {Kawabata}}, \bibinfo {author} {\bibfnamefont {K.}~\bibnamefont {Shiozaki}}, \bibinfo {author} {\bibfnamefont {M.}~\bibnamefont {Ueda}},\ and\ \bibinfo {author} {\bibfnamefont {M.}~\bibnamefont {Sato}},\ }\bibfield  {title} {\bibinfo {title} {Symmetry and topology in non-hermitian physics},\ }\href@noop {} {\bibfield  {journal} {\bibinfo  {journal} {Physical Review X}\ }\textbf {\bibinfo {volume} {9}},\ \bibinfo {pages} {041015} (\bibinfo {year} {2019})}\BibitemShut {NoStop}%
\bibitem [{\citenamefont {Okuma}\ \emph {et~al.}(2020)\citenamefont {Okuma}, \citenamefont {Kawabata}, \citenamefont {Shiozaki},\ and\ \citenamefont {Sato}}]{okuma2020topological}%
  \BibitemOpen
  \bibfield  {author} {\bibinfo {author} {\bibfnamefont {N.}~\bibnamefont {Okuma}}, \bibinfo {author} {\bibfnamefont {K.}~\bibnamefont {Kawabata}}, \bibinfo {author} {\bibfnamefont {K.}~\bibnamefont {Shiozaki}},\ and\ \bibinfo {author} {\bibfnamefont {M.}~\bibnamefont {Sato}},\ }\bibfield  {title} {\bibinfo {title} {Topological origin of non-hermitian skin effects},\ }\href@noop {} {\bibfield  {journal} {\bibinfo  {journal} {Physical review letters}\ }\textbf {\bibinfo {volume} {124}},\ \bibinfo {pages} {086801} (\bibinfo {year} {2020})}\BibitemShut {NoStop}%
\bibitem [{\citenamefont {Ashida}\ \emph {et~al.}(2020)\citenamefont {Ashida}, \citenamefont {Gong},\ and\ \citenamefont {Ueda}}]{AshidaGongUeda2020}%
  \BibitemOpen
  \bibfield  {author} {\bibinfo {author} {\bibfnamefont {Y.}~\bibnamefont {Ashida}}, \bibinfo {author} {\bibfnamefont {Z.}~\bibnamefont {Gong}},\ and\ \bibinfo {author} {\bibfnamefont {M.}~\bibnamefont {Ueda}},\ }\bibfield  {title} {\bibinfo {title} {Non-hermitian physics},\ }\bibfield  {booktitle} {\emph {\bibinfo {booktitle} {Advances in Physics}},\ }\href {https://doi.org/10.1080/00018732.2021.1876991} {\bibfield  {journal} {\bibinfo  {journal} {Advances in Physics}\ }\textbf {\bibinfo {volume} {69}},\ \bibinfo {pages} {249} (\bibinfo {year} {2020})}\BibitemShut {NoStop}%
\bibitem [{\citenamefont {Matsumoto}\ \emph {et~al.}(2020)\citenamefont {Matsumoto}, \citenamefont {Kawabata}, \citenamefont {Ashida}, \citenamefont {Furukawa},\ and\ \citenamefont {Ueda}}]{MatsumotoKawabataAshidaEtAl2020}%
  \BibitemOpen
  \bibfield  {author} {\bibinfo {author} {\bibfnamefont {N.}~\bibnamefont {Matsumoto}}, \bibinfo {author} {\bibfnamefont {K.}~\bibnamefont {Kawabata}}, \bibinfo {author} {\bibfnamefont {Y.}~\bibnamefont {Ashida}}, \bibinfo {author} {\bibfnamefont {S.}~\bibnamefont {Furukawa}},\ and\ \bibinfo {author} {\bibfnamefont {M.}~\bibnamefont {Ueda}},\ }\bibfield  {title} {\bibinfo {title} {Continuous phase transition without gap closing in non-hermitian quantum many-body systems},\ }\href {https://doi.org/10.1103/PhysRevLett.125.260601} {\bibfield  {journal} {\bibinfo  {journal} {Phys. Rev. Lett.}\ }\textbf {\bibinfo {volume} {125}},\ \bibinfo {pages} {260601} (\bibinfo {year} {2020})}\BibitemShut {NoStop}%
\bibitem [{\citenamefont {Bergholtz}\ \emph {et~al.}(2021)\citenamefont {Bergholtz}, \citenamefont {Budich},\ and\ \citenamefont {Kunst}}]{bergholtz2021exceptional}%
  \BibitemOpen
  \bibfield  {author} {\bibinfo {author} {\bibfnamefont {E.~J.}\ \bibnamefont {Bergholtz}}, \bibinfo {author} {\bibfnamefont {J.~C.}\ \bibnamefont {Budich}},\ and\ \bibinfo {author} {\bibfnamefont {F.~K.}\ \bibnamefont {Kunst}},\ }\bibfield  {title} {\bibinfo {title} {Exceptional topology of non-hermitian systems},\ }\href@noop {} {\bibfield  {journal} {\bibinfo  {journal} {Reviews of Modern Physics}\ }\textbf {\bibinfo {volume} {93}},\ \bibinfo {pages} {015005} (\bibinfo {year} {2021})}\BibitemShut {NoStop}%
\bibitem [{\citenamefont {Ding}\ \emph {et~al.}(2022)\citenamefont {Ding}, \citenamefont {Fang},\ and\ \citenamefont {Ma}}]{DingFangMa2022}%
  \BibitemOpen
  \bibfield  {author} {\bibinfo {author} {\bibfnamefont {K.}~\bibnamefont {Ding}}, \bibinfo {author} {\bibfnamefont {C.}~\bibnamefont {Fang}},\ and\ \bibinfo {author} {\bibfnamefont {G.}~\bibnamefont {Ma}},\ }\bibfield  {title} {\bibinfo {title} {Non-hermitian topology and exceptional-point geometries},\ }\href {https://doi.org/10.1038/s42254-022-00516-5} {\bibfield  {journal} {\bibinfo  {journal} {Nature Reviews Physics}\ }\textbf {\bibinfo {volume} {4}},\ \bibinfo {pages} {745–760} (\bibinfo {year} {2022})}\BibitemShut {NoStop}%
\bibitem [{\citenamefont {Chen}\ \emph {et~al.}(2023{\natexlab{a}})\citenamefont {Chen}, \citenamefont {Song},\ and\ \citenamefont {Lado}}]{ChenSongLado2023}%
  \BibitemOpen
  \bibfield  {author} {\bibinfo {author} {\bibfnamefont {G.}~\bibnamefont {Chen}}, \bibinfo {author} {\bibfnamefont {F.}~\bibnamefont {Song}},\ and\ \bibinfo {author} {\bibfnamefont {J.~L.}\ \bibnamefont {Lado}},\ }\bibfield  {title} {\bibinfo {title} {Topological spin excitations in non-hermitian spin chains with a generalized kernel polynomial algorithm},\ }\href {https://doi.org/10.1103/PhysRevLett.130.100401} {\bibfield  {journal} {\bibinfo  {journal} {Phys. Rev. Lett.}\ }\textbf {\bibinfo {volume} {130}},\ \bibinfo {pages} {100401} (\bibinfo {year} {2023}{\natexlab{a}})}\BibitemShut {NoStop}%
\bibitem [{\citenamefont {Zheng}\ \emph {et~al.}(2024)\citenamefont {Zheng}, \citenamefont {Qiao}, \citenamefont {Wang}, \citenamefont {Cao},\ and\ \citenamefont {Chen}}]{ZhengQiaoWangEtAl2024}%
  \BibitemOpen
  \bibfield  {author} {\bibinfo {author} {\bibfnamefont {M.}~\bibnamefont {Zheng}}, \bibinfo {author} {\bibfnamefont {Y.}~\bibnamefont {Qiao}}, \bibinfo {author} {\bibfnamefont {Y.}~\bibnamefont {Wang}}, \bibinfo {author} {\bibfnamefont {J.}~\bibnamefont {Cao}},\ and\ \bibinfo {author} {\bibfnamefont {S.}~\bibnamefont {Chen}},\ }\bibfield  {title} {\bibinfo {title} {Exact solution of the bose-hubbard model with unidirectional hopping},\ }\href {https://doi.org/10.1103/PhysRevLett.132.086502} {\bibfield  {journal} {\bibinfo  {journal} {Phys. Rev. Lett.}\ }\textbf {\bibinfo {volume} {132}},\ \bibinfo {pages} {086502} (\bibinfo {year} {2024})}\BibitemShut {NoStop}%
\bibitem [{\citenamefont {Shen}\ \emph {et~al.}(2024)\citenamefont {Shen}, \citenamefont {Lu}, \citenamefont {Lado},\ and\ \citenamefont {Trif}}]{ShenLuLadoEtAl2024}%
  \BibitemOpen
  \bibfield  {author} {\bibinfo {author} {\bibfnamefont {P.-X.}\ \bibnamefont {Shen}}, \bibinfo {author} {\bibfnamefont {Z.}~\bibnamefont {Lu}}, \bibinfo {author} {\bibfnamefont {J.~L.}\ \bibnamefont {Lado}},\ and\ \bibinfo {author} {\bibfnamefont {M.}~\bibnamefont {Trif}},\ }\bibfield  {title} {\bibinfo {title} {Non-hermitian fermi-dirac distribution in persistent current transport},\ }\href {https://doi.org/10.1103/PhysRevLett.133.086301} {\bibfield  {journal} {\bibinfo  {journal} {Phys. Rev. Lett.}\ }\textbf {\bibinfo {volume} {133}},\ \bibinfo {pages} {086301} (\bibinfo {year} {2024})}\BibitemShut {NoStop}%
\bibitem [{\citenamefont {Batchelor}(2000)}]{Batchelor_2000}%
  \BibitemOpen
  \bibfield  {author} {\bibinfo {author} {\bibfnamefont {G.~K.}\ \bibnamefont {Batchelor}},\ }\bibinfo {title} {Frontmatter},\ in\ \href@noop {} {\emph {\bibinfo {booktitle} {An Introduction to Fluid Dynamics}}},\ \bibinfo {series and number} {Cambridge Mathematical Library}\ (\bibinfo  {publisher} {Cambridge University Press},\ \bibinfo {year} {2000})\ p.\ \bibinfo {pages} {i–iv}\BibitemShut {NoStop}%
\bibitem [{\citenamefont {Evans}(2010)}]{evans2010partial}%
  \BibitemOpen
  \bibfield  {author} {\bibinfo {author} {\bibfnamefont {L.~C.}\ \bibnamefont {Evans}},\ }\href@noop {} {\emph {\bibinfo {title} {Partial differential equations}}},\ Vol.~\bibinfo {volume} {19}\ (\bibinfo  {publisher} {American Mathematical Soc.},\ \bibinfo {year} {2010})\BibitemShut {NoStop}%
\bibitem [{\citenamefont {Widder}(1976)}]{widder1976heat}%
  \BibitemOpen
  \bibfield  {author} {\bibinfo {author} {\bibfnamefont {D.~V.}\ \bibnamefont {Widder}},\ }\href@noop {} {\emph {\bibinfo {title} {The heat equation}}},\ Vol.~\bibinfo {volume} {67}\ (\bibinfo  {publisher} {Academic Press},\ \bibinfo {year} {1976})\BibitemShut {NoStop}%
\bibitem [{\citenamefont {Cannon}(1984)}]{cannon1984one}%
  \BibitemOpen
  \bibfield  {author} {\bibinfo {author} {\bibfnamefont {J.~R.}\ \bibnamefont {Cannon}},\ }\href@noop {} {\emph {\bibinfo {title} {The one-dimensional heat equation}}},\ \bibinfo {number} {23}\ (\bibinfo  {publisher} {Cambridge University Press},\ \bibinfo {year} {1984})\BibitemShut {NoStop}%
\bibitem [{\citenamefont {Hundsdorfer}\ and\ \citenamefont {Verwer}(2003)}]{hundsdorfer2003numerical}%
  \BibitemOpen
  \bibfield  {author} {\bibinfo {author} {\bibfnamefont {W.~H.}\ \bibnamefont {Hundsdorfer}}\ and\ \bibinfo {author} {\bibfnamefont {J.~G.}\ \bibnamefont {Verwer}},\ }\href@noop {} {\emph {\bibinfo {title} {Numerical solution of time-dependent advection-diffusion-reaction equations}}},\ Vol.~\bibinfo {volume} {33}\ (\bibinfo  {publisher} {Springer},\ \bibinfo {year} {2003})\BibitemShut {NoStop}%
\bibitem [{\citenamefont {Thambynayagam}(2011)}]{thambynayagam2011diffusion}%
  \BibitemOpen
  \bibfield  {author} {\bibinfo {author} {\bibfnamefont {R.~K.~M.}\ \bibnamefont {Thambynayagam}},\ }\href@noop {} {\emph {\bibinfo {title} {The diffusion handbook: applied solutions for engineers}}}\ (\bibinfo  {publisher} {McGraw Hill Professional},\ \bibinfo {year} {2011})\BibitemShut {NoStop}%
\bibitem [{\citenamefont {Gily{\'e}n}\ \emph {et~al.}(2019)\citenamefont {Gily{\'e}n}, \citenamefont {Su}, \citenamefont {Low},\ and\ \citenamefont {Wiebe}}]{gilyen2019quantum}%
  \BibitemOpen
  \bibfield  {author} {\bibinfo {author} {\bibfnamefont {A.}~\bibnamefont {Gily{\'e}n}}, \bibinfo {author} {\bibfnamefont {Y.}~\bibnamefont {Su}}, \bibinfo {author} {\bibfnamefont {G.~H.}\ \bibnamefont {Low}},\ and\ \bibinfo {author} {\bibfnamefont {N.}~\bibnamefont {Wiebe}},\ }\bibfield  {title} {\bibinfo {title} {Quantum singular value transformation and beyond: Exponential improvements for quantum matrix arithmetics},\ }in\ \href {https://doi.org/10.1145/3313276.3316366} {\emph {\bibinfo {booktitle} {Proceedings of the 51st {{Annual ACM SIGACT Symposium}} on {{Theory}} of {{Computing}}}}}\ (\bibinfo {year} {2019})\ pp.\ \bibinfo {pages} {193--204},\ \Eprint {https://arxiv.org/abs/1806.01838} {arxiv:1806.01838 [quant-ph]} \BibitemShut {NoStop}%
\bibitem [{\citenamefont {Harrow}\ \emph {et~al.}(2009)\citenamefont {Harrow}, \citenamefont {Hassidim},\ and\ \citenamefont {Lloyd}}]{HarrowHassidimLloyd2009}%
  \BibitemOpen
  \bibfield  {author} {\bibinfo {author} {\bibfnamefont {A.~W.}\ \bibnamefont {Harrow}}, \bibinfo {author} {\bibfnamefont {A.}~\bibnamefont {Hassidim}},\ and\ \bibinfo {author} {\bibfnamefont {S.}~\bibnamefont {Lloyd}},\ }\bibfield  {title} {\bibinfo {title} {Quantum algorithm for linear systems of equations},\ }\href {https://doi.org/10.1103/PhysRevLett.103.150502} {\bibfield  {journal} {\bibinfo  {journal} {Phys. Rev. Lett.}\ }\textbf {\bibinfo {volume} {103}},\ \bibinfo {pages} {150502} (\bibinfo {year} {2009})}\BibitemShut {NoStop}%
\bibitem [{\citenamefont {Ambainis}(2012)}]{Ambainis2012}%
  \BibitemOpen
  \bibfield  {author} {\bibinfo {author} {\bibfnamefont {A.}~\bibnamefont {Ambainis}},\ }\bibfield  {title} {\bibinfo {title} {Variable time amplitude amplification and quantum algorithms for linear algebra problems},\ }in\ \href {https://doi.org/10.4230/LIPIcs.STACS.2012.636} {\emph {\bibinfo {booktitle} {STACS'12 (29th Symposium on Theoretical Aspects of Computer Science)}}},\ Vol.~\bibinfo {volume} {14}\ (\bibinfo {year} {2012})\ pp.\ \bibinfo {pages} {636--647}\BibitemShut {NoStop}%
\bibitem [{\citenamefont {Childs}\ \emph {et~al.}(2017{\natexlab{a}})\citenamefont {Childs}, \citenamefont {Kothari},\ and\ \citenamefont {Somma}}]{ChildsKothariSomma2017}%
  \BibitemOpen
  \bibfield  {author} {\bibinfo {author} {\bibfnamefont {A.~M.}\ \bibnamefont {Childs}}, \bibinfo {author} {\bibfnamefont {R.}~\bibnamefont {Kothari}},\ and\ \bibinfo {author} {\bibfnamefont {R.~D.}\ \bibnamefont {Somma}},\ }\bibfield  {title} {\bibinfo {title} {Quantum algorithm for systems of linear equations with exponentially improved dependence on precision},\ }\href {https://doi.org/10.1137/16M1087072} {\bibfield  {journal} {\bibinfo  {journal} {SIAM J. Comput.}\ }\textbf {\bibinfo {volume} {46}},\ \bibinfo {pages} {1920} (\bibinfo {year} {2017}{\natexlab{a}})}\BibitemShut {NoStop}%
\bibitem [{\citenamefont {Suba{\c{s}}{\i}}\ \emph {et~al.}(2019)\citenamefont {Suba{\c{s}}{\i}}, \citenamefont {Somma},\ and\ \citenamefont {Orsucci}}]{SubasiSommaOrsucci2019}%
  \BibitemOpen
  \bibfield  {author} {\bibinfo {author} {\bibfnamefont {Y.}~\bibnamefont {Suba{\c{s}}{\i}}}, \bibinfo {author} {\bibfnamefont {R.~D.}\ \bibnamefont {Somma}},\ and\ \bibinfo {author} {\bibfnamefont {D.}~\bibnamefont {Orsucci}},\ }\bibfield  {title} {\bibinfo {title} {Quantum algorithms for systems of linear equations inspired by adiabatic quantum computing},\ }\href {https://doi.org/10.1103/PhysRevLett.122.060504} {\bibfield  {journal} {\bibinfo  {journal} {Phys. Rev. Lett.}\ }\textbf {\bibinfo {volume} {122}},\ \bibinfo {pages} {060504} (\bibinfo {year} {2019})}\BibitemShut {NoStop}%
\bibitem [{\citenamefont {Lin}\ and\ \citenamefont {Tong}(2020)}]{LinTong2020}%
  \BibitemOpen
  \bibfield  {author} {\bibinfo {author} {\bibfnamefont {L.}~\bibnamefont {Lin}}\ and\ \bibinfo {author} {\bibfnamefont {Y.}~\bibnamefont {Tong}},\ }\bibfield  {title} {\bibinfo {title} {Optimal polynomial based quantum eigenstate filtering with application to solving quantum linear systems},\ }\href {https://doi.org/10.22331/q-2020-11-11-361} {\bibfield  {journal} {\bibinfo  {journal} {Quantum}\ }\textbf {\bibinfo {volume} {4}},\ \bibinfo {pages} {361} (\bibinfo {year} {2020})}\BibitemShut {NoStop}%
\bibitem [{\citenamefont {An}\ and\ \citenamefont {Lin}(2022)}]{AnLin2022}%
  \BibitemOpen
  \bibfield  {author} {\bibinfo {author} {\bibfnamefont {D.}~\bibnamefont {An}}\ and\ \bibinfo {author} {\bibfnamefont {L.}~\bibnamefont {Lin}},\ }\bibfield  {title} {\bibinfo {title} {Quantum linear system solver based on time-optimal adiabatic quantum computing and quantum approximate optimization algorithm},\ }\bibfield  {journal} {\bibinfo  {journal} {ACM Transactions on Quantum Computing}\ }\textbf {\bibinfo {volume} {3}},\ \href {https://doi.org/10.1145/3498331} {10.1145/3498331} (\bibinfo {year} {2022})\BibitemShut {NoStop}%
\bibitem [{\citenamefont {Costa}\ \emph {et~al.}(2022)\citenamefont {Costa}, \citenamefont {An}, \citenamefont {Sanders}, \citenamefont {Su}, \citenamefont {Babbush},\ and\ \citenamefont {Berry}}]{CostaAnYuvalEtAl2022}%
  \BibitemOpen
  \bibfield  {author} {\bibinfo {author} {\bibfnamefont {P.~C.}\ \bibnamefont {Costa}}, \bibinfo {author} {\bibfnamefont {D.}~\bibnamefont {An}}, \bibinfo {author} {\bibfnamefont {Y.~R.}\ \bibnamefont {Sanders}}, \bibinfo {author} {\bibfnamefont {Y.}~\bibnamefont {Su}}, \bibinfo {author} {\bibfnamefont {R.}~\bibnamefont {Babbush}},\ and\ \bibinfo {author} {\bibfnamefont {D.~W.}\ \bibnamefont {Berry}},\ }\bibfield  {title} {\bibinfo {title} {Optimal scaling quantum linear-systems solver via discrete adiabatic theorem},\ }\href {https://doi.org/10.1103/PRXQuantum.3.040303} {\bibfield  {journal} {\bibinfo  {journal} {PRX Quantum}\ }\textbf {\bibinfo {volume} {3}},\ \bibinfo {pages} {040303} (\bibinfo {year} {2022})}\BibitemShut {NoStop}%
\bibitem [{\citenamefont {Dalzell}(2024)}]{Dalzell2024}%
  \BibitemOpen
  \bibfield  {author} {\bibinfo {author} {\bibfnamefont {A.~M.}\ \bibnamefont {Dalzell}},\ }\href {https://arxiv.org/abs/2406.12086} {\bibinfo {title} {A shortcut to an optimal quantum linear system solver}} (\bibinfo {year} {2024}),\ \Eprint {https://arxiv.org/abs/2406.12086} {arXiv:2406.12086 [quant-ph]} \BibitemShut {NoStop}%
\bibitem [{\citenamefont {Lindblad}(1976)}]{lindblad1976generators}%
  \BibitemOpen
  \bibfield  {author} {\bibinfo {author} {\bibfnamefont {G.}~\bibnamefont {Lindblad}},\ }\bibfield  {title} {\bibinfo {title} {On the generators of quantum dynamical semigroups},\ }\href@noop {} {\bibfield  {journal} {\bibinfo  {journal} {Communications in mathematical physics}\ }\textbf {\bibinfo {volume} {48}},\ \bibinfo {pages} {119} (\bibinfo {year} {1976})}\BibitemShut {NoStop}%
\bibitem [{\citenamefont {Verstraete}\ \emph {et~al.}(2009)\citenamefont {Verstraete}, \citenamefont {Wolf},\ and\ \citenamefont {Ignacio~Cirac}}]{verstraete2009quantum}%
  \BibitemOpen
  \bibfield  {author} {\bibinfo {author} {\bibfnamefont {F.}~\bibnamefont {Verstraete}}, \bibinfo {author} {\bibfnamefont {M.~M.}\ \bibnamefont {Wolf}},\ and\ \bibinfo {author} {\bibfnamefont {J.}~\bibnamefont {Ignacio~Cirac}},\ }\bibfield  {title} {\bibinfo {title} {Quantum computation and quantum-state engineering driven by dissipation},\ }\href@noop {} {\bibfield  {journal} {\bibinfo  {journal} {Nature physics}\ }\textbf {\bibinfo {volume} {5}},\ \bibinfo {pages} {633} (\bibinfo {year} {2009})}\BibitemShut {NoStop}%
\bibitem [{\citenamefont {Chen}\ \emph {et~al.}(2024)\citenamefont {Chen}, \citenamefont {Huang}, \citenamefont {Preskill},\ and\ \citenamefont {Zhou}}]{chen2024local}%
  \BibitemOpen
  \bibfield  {author} {\bibinfo {author} {\bibfnamefont {C.-F.}\ \bibnamefont {Chen}}, \bibinfo {author} {\bibfnamefont {H.-Y.}\ \bibnamefont {Huang}}, \bibinfo {author} {\bibfnamefont {J.}~\bibnamefont {Preskill}},\ and\ \bibinfo {author} {\bibfnamefont {L.}~\bibnamefont {Zhou}},\ }\bibfield  {title} {\bibinfo {title} {Local minima in quantum systems},\ }in\ \href@noop {} {\emph {\bibinfo {booktitle} {Proceedings of the 56th Annual ACM Symposium on Theory of Computing}}}\ (\bibinfo {year} {2024})\ pp.\ \bibinfo {pages} {1323--1330}\BibitemShut {NoStop}%
\bibitem [{\citenamefont {Guo}\ \emph {et~al.}(2024)\citenamefont {Guo}, \citenamefont {Hart}, \citenamefont {Chen}, \citenamefont {Friedman},\ and\ \citenamefont {Lucas}}]{guo2024designing}%
  \BibitemOpen
  \bibfield  {author} {\bibinfo {author} {\bibfnamefont {J.}~\bibnamefont {Guo}}, \bibinfo {author} {\bibfnamefont {O.}~\bibnamefont {Hart}}, \bibinfo {author} {\bibfnamefont {C.-F.}\ \bibnamefont {Chen}}, \bibinfo {author} {\bibfnamefont {A.~J.}\ \bibnamefont {Friedman}},\ and\ \bibinfo {author} {\bibfnamefont {A.}~\bibnamefont {Lucas}},\ }\bibfield  {title} {\bibinfo {title} {Designing open quantum systems with known steady states: Davies generators and beyond},\ }\href@noop {} {\bibfield  {journal} {\bibinfo  {journal} {arXiv preprint arXiv:2404.14538}\ } (\bibinfo {year} {2024})}\BibitemShut {NoStop}%
\bibitem [{\citenamefont {Chen}\ \emph {et~al.}(2023{\natexlab{b}})\citenamefont {Chen}, \citenamefont {Kastoryano}, \citenamefont {Brand{\~a}o},\ and\ \citenamefont {Gily{\'e}n}}]{chen2023quantum}%
  \BibitemOpen
  \bibfield  {author} {\bibinfo {author} {\bibfnamefont {C.-F.}\ \bibnamefont {Chen}}, \bibinfo {author} {\bibfnamefont {M.~J.}\ \bibnamefont {Kastoryano}}, \bibinfo {author} {\bibfnamefont {F.~G.}\ \bibnamefont {Brand{\~a}o}},\ and\ \bibinfo {author} {\bibfnamefont {A.}~\bibnamefont {Gily{\'e}n}},\ }\bibfield  {title} {\bibinfo {title} {Quantum thermal state preparation},\ }\href@noop {} {\bibfield  {journal} {\bibinfo  {journal} {arXiv preprint arXiv:2303.18224}\ } (\bibinfo {year} {2023}{\natexlab{b}})}\BibitemShut {NoStop}%
\bibitem [{\citenamefont {Chen}\ \emph {et~al.}(2023{\natexlab{c}})\citenamefont {Chen}, \citenamefont {Kastoryano},\ and\ \citenamefont {Gily{\'e}n}}]{chen2023efficient}%
  \BibitemOpen
  \bibfield  {author} {\bibinfo {author} {\bibfnamefont {C.-F.}\ \bibnamefont {Chen}}, \bibinfo {author} {\bibfnamefont {M.~J.}\ \bibnamefont {Kastoryano}},\ and\ \bibinfo {author} {\bibfnamefont {A.}~\bibnamefont {Gily{\'e}n}},\ }\bibfield  {title} {\bibinfo {title} {An efficient and exact noncommutative quantum gibbs sampler},\ }\href@noop {} {\bibfield  {journal} {\bibinfo  {journal} {arXiv preprint arXiv:2311.09207}\ } (\bibinfo {year} {2023}{\natexlab{c}})}\BibitemShut {NoStop}%
\bibitem [{\citenamefont {Ding}\ \emph {et~al.}(2024{\natexlab{a}})\citenamefont {Ding}, \citenamefont {Li},\ and\ \citenamefont {Lin}}]{ding2024efficient}%
  \BibitemOpen
  \bibfield  {author} {\bibinfo {author} {\bibfnamefont {Z.}~\bibnamefont {Ding}}, \bibinfo {author} {\bibfnamefont {B.}~\bibnamefont {Li}},\ and\ \bibinfo {author} {\bibfnamefont {L.}~\bibnamefont {Lin}},\ }\bibfield  {title} {\bibinfo {title} {Efficient quantum gibbs samplers with kubo--martin--schwinger detailed balance condition},\ }\href@noop {} {\bibfield  {journal} {\bibinfo  {journal} {arXiv preprint arXiv:2404.05998}\ } (\bibinfo {year} {2024}{\natexlab{a}})}\BibitemShut {NoStop}%
\bibitem [{\citenamefont {Gily{\'e}n}\ \emph {et~al.}(2024)\citenamefont {Gily{\'e}n}, \citenamefont {Chen}, \citenamefont {Doriguello},\ and\ \citenamefont {Kastoryano}}]{gilyen2024quantum}%
  \BibitemOpen
  \bibfield  {author} {\bibinfo {author} {\bibfnamefont {A.}~\bibnamefont {Gily{\'e}n}}, \bibinfo {author} {\bibfnamefont {C.-F.}\ \bibnamefont {Chen}}, \bibinfo {author} {\bibfnamefont {J.~F.}\ \bibnamefont {Doriguello}},\ and\ \bibinfo {author} {\bibfnamefont {M.~J.}\ \bibnamefont {Kastoryano}},\ }\bibfield  {title} {\bibinfo {title} {Quantum generalizations of glauber and metropolis dynamics},\ }\href@noop {} {\bibfield  {journal} {\bibinfo  {journal} {arXiv preprint arXiv:2405.20322}\ } (\bibinfo {year} {2024})}\BibitemShut {NoStop}%
\bibitem [{\citenamefont {Ding}\ \emph {et~al.}(2024{\natexlab{b}})\citenamefont {Ding}, \citenamefont {Chen},\ and\ \citenamefont {Lin}}]{ding2024single}%
  \BibitemOpen
  \bibfield  {author} {\bibinfo {author} {\bibfnamefont {Z.}~\bibnamefont {Ding}}, \bibinfo {author} {\bibfnamefont {C.-F.}\ \bibnamefont {Chen}},\ and\ \bibinfo {author} {\bibfnamefont {L.}~\bibnamefont {Lin}},\ }\bibfield  {title} {\bibinfo {title} {Single-ancilla ground state preparation via lindbladians},\ }\href@noop {} {\bibfield  {journal} {\bibinfo  {journal} {Physical Review Research}\ }\textbf {\bibinfo {volume} {6}},\ \bibinfo {pages} {033147} (\bibinfo {year} {2024}{\natexlab{b}})}\BibitemShut {NoStop}%
\bibitem [{\citenamefont {Shang}\ \emph {et~al.}(2024)\citenamefont {Shang}, \citenamefont {Chen}, \citenamefont {Chen}, \citenamefont {Lu},\ and\ \citenamefont {Pan}}]{shang2024polynomial}%
  \BibitemOpen
  \bibfield  {author} {\bibinfo {author} {\bibfnamefont {Z.-X.}\ \bibnamefont {Shang}}, \bibinfo {author} {\bibfnamefont {Z.-H.}\ \bibnamefont {Chen}}, \bibinfo {author} {\bibfnamefont {M.-C.}\ \bibnamefont {Chen}}, \bibinfo {author} {\bibfnamefont {C.-Y.}\ \bibnamefont {Lu}},\ and\ \bibinfo {author} {\bibfnamefont {J.-W.}\ \bibnamefont {Pan}},\ }\bibfield  {title} {\bibinfo {title} {A polynomial-time quantum algorithm for solving the ground states of a class of classically hard hamiltonians},\ }\href@noop {} {\bibfield  {journal} {\bibinfo  {journal} {arXiv preprint arXiv:2401.13946}\ } (\bibinfo {year} {2024})}\BibitemShut {NoStop}%
\bibitem [{\citenamefont {Kliesch}\ \emph {et~al.}(2011)\citenamefont {Kliesch}, \citenamefont {Barthel}, \citenamefont {Gogolin}, \citenamefont {Kastoryano},\ and\ \citenamefont {Eisert}}]{kliesch2011dissipative}%
  \BibitemOpen
  \bibfield  {author} {\bibinfo {author} {\bibfnamefont {M.}~\bibnamefont {Kliesch}}, \bibinfo {author} {\bibfnamefont {T.}~\bibnamefont {Barthel}}, \bibinfo {author} {\bibfnamefont {C.}~\bibnamefont {Gogolin}}, \bibinfo {author} {\bibfnamefont {M.}~\bibnamefont {Kastoryano}},\ and\ \bibinfo {author} {\bibfnamefont {J.}~\bibnamefont {Eisert}},\ }\bibfield  {title} {\bibinfo {title} {Dissipative quantum church-turing theorem},\ }\href@noop {} {\bibfield  {journal} {\bibinfo  {journal} {Physical review letters}\ }\textbf {\bibinfo {volume} {107}},\ \bibinfo {pages} {120501} (\bibinfo {year} {2011})}\BibitemShut {NoStop}%
\bibitem [{\citenamefont {Cleve}\ and\ \citenamefont {Wang}(2016)}]{cleve2016efficient}%
  \BibitemOpen
  \bibfield  {author} {\bibinfo {author} {\bibfnamefont {R.}~\bibnamefont {Cleve}}\ and\ \bibinfo {author} {\bibfnamefont {C.}~\bibnamefont {Wang}},\ }\bibfield  {title} {\bibinfo {title} {Efficient quantum algorithms for simulating lindblad evolution},\ }\href@noop {} {\bibfield  {journal} {\bibinfo  {journal} {arXiv preprint arXiv:1612.09512}\ } (\bibinfo {year} {2016})}\BibitemShut {NoStop}%
\bibitem [{\citenamefont {Childs}\ and\ \citenamefont {Li}(2017)}]{childs2017efficient}%
  \BibitemOpen
  \bibfield  {author} {\bibinfo {author} {\bibfnamefont {A.~M.}\ \bibnamefont {Childs}}\ and\ \bibinfo {author} {\bibfnamefont {T.}~\bibnamefont {Li}},\ }\bibfield  {title} {\bibinfo {title} {Efficient simulation of sparse markovian quantum dynamics},\ }\href@noop {} {\bibfield  {journal} {\bibinfo  {journal} {Quantum Info. Comput.}\ }\textbf {\bibinfo {volume} {17}},\ \bibinfo {pages} {901–947} (\bibinfo {year} {2017})}\BibitemShut {NoStop}%
\bibitem [{\citenamefont {Schlimgen}\ \emph {et~al.}(2021)\citenamefont {Schlimgen}, \citenamefont {Head-Marsden}, \citenamefont {Sager}, \citenamefont {Narang},\ and\ \citenamefont {Mazziotti}}]{schlimgen2021quantum}%
  \BibitemOpen
  \bibfield  {author} {\bibinfo {author} {\bibfnamefont {A.~W.}\ \bibnamefont {Schlimgen}}, \bibinfo {author} {\bibfnamefont {K.}~\bibnamefont {Head-Marsden}}, \bibinfo {author} {\bibfnamefont {L.~M.}\ \bibnamefont {Sager}}, \bibinfo {author} {\bibfnamefont {P.}~\bibnamefont {Narang}},\ and\ \bibinfo {author} {\bibfnamefont {D.~A.}\ \bibnamefont {Mazziotti}},\ }\bibfield  {title} {\bibinfo {title} {Quantum simulation of open quantum systems using a unitary decomposition of operators},\ }\href@noop {} {\bibfield  {journal} {\bibinfo  {journal} {Physical Review Letters}\ }\textbf {\bibinfo {volume} {127}},\ \bibinfo {pages} {270503} (\bibinfo {year} {2021})}\BibitemShut {NoStop}%
\bibitem [{\citenamefont {Li}\ and\ \citenamefont {Wang}(2023)}]{li2023opensystem}%
  \BibitemOpen
  \bibfield  {author} {\bibinfo {author} {\bibfnamefont {X.}~\bibnamefont {Li}}\ and\ \bibinfo {author} {\bibfnamefont {C.}~\bibnamefont {Wang}},\ }\bibfield  {title} {\bibinfo {title} {{Simulating Markovian Open Quantum Systems Using Higher-Order Series Expansion}},\ }in\ \href {https://doi.org/10.4230/LIPIcs.ICALP.2023.87} {\emph {\bibinfo {booktitle} {50th International Colloquium on Automata, Languages, and Programming (ICALP 2023)}}},\ \bibinfo {series} {Leibniz International Proceedings in Informatics (LIPIcs)}, Vol.\ \bibinfo {volume} {261},\ \bibinfo {editor} {edited by\ \bibinfo {editor} {\bibfnamefont {K.}~\bibnamefont {Etessami}}, \bibinfo {editor} {\bibfnamefont {U.}~\bibnamefont {Feige}},\ and\ \bibinfo {editor} {\bibfnamefont {G.}~\bibnamefont {Puppis}}}\ (\bibinfo  {publisher} {Schloss Dagstuhl -- Leibniz-Zentrum f{\"u}r Informatik},\ \bibinfo {address} {Dagstuhl, Germany},\ \bibinfo {year} {2023})\ pp.\ \bibinfo {pages} {87:1--87:20}\BibitemShut {NoStop}%
\bibitem [{\citenamefont {Pocrnic}\ \emph {et~al.}(2023)\citenamefont {Pocrnic}, \citenamefont {Segal},\ and\ \citenamefont {Wiebe}}]{pocrnic2023quantum}%
  \BibitemOpen
  \bibfield  {author} {\bibinfo {author} {\bibfnamefont {M.}~\bibnamefont {Pocrnic}}, \bibinfo {author} {\bibfnamefont {D.}~\bibnamefont {Segal}},\ and\ \bibinfo {author} {\bibfnamefont {N.}~\bibnamefont {Wiebe}},\ }\bibfield  {title} {\bibinfo {title} {Quantum simulation of lindbladian dynamics via repeated interactions},\ }\href@noop {} {\bibfield  {journal} {\bibinfo  {journal} {arXiv preprint arXiv:2312.05371}\ } (\bibinfo {year} {2023})}\BibitemShut {NoStop}%
\bibitem [{\citenamefont {He}\ \emph {et~al.}(2024)\citenamefont {He}, \citenamefont {Li}, \citenamefont {Li}, \citenamefont {Li}, \citenamefont {Wang},\ and\ \citenamefont {Wang}}]{he2024efficientoptimalcontrolopen}%
  \BibitemOpen
  \bibfield  {author} {\bibinfo {author} {\bibfnamefont {W.}~\bibnamefont {He}}, \bibinfo {author} {\bibfnamefont {T.}~\bibnamefont {Li}}, \bibinfo {author} {\bibfnamefont {X.}~\bibnamefont {Li}}, \bibinfo {author} {\bibfnamefont {Z.}~\bibnamefont {Li}}, \bibinfo {author} {\bibfnamefont {C.}~\bibnamefont {Wang}},\ and\ \bibinfo {author} {\bibfnamefont {K.}~\bibnamefont {Wang}},\ }\href {https://arxiv.org/abs/2405.19245} {\bibinfo {title} {Efficient optimal control of open quantum systems}} (\bibinfo {year} {2024}),\ \Eprint {https://arxiv.org/abs/2405.19245} {arXiv:2405.19245 [quant-ph]} \BibitemShut {NoStop}%
\bibitem [{\citenamefont {Ding}\ \emph {et~al.}(2024{\natexlab{c}})\citenamefont {Ding}, \citenamefont {Li},\ and\ \citenamefont {Lin}}]{ding2024simulating}%
  \BibitemOpen
  \bibfield  {author} {\bibinfo {author} {\bibfnamefont {Z.}~\bibnamefont {Ding}}, \bibinfo {author} {\bibfnamefont {X.}~\bibnamefont {Li}},\ and\ \bibinfo {author} {\bibfnamefont {L.}~\bibnamefont {Lin}},\ }\bibfield  {title} {\bibinfo {title} {Simulating open quantum systems using hamiltonian simulations},\ }\href {https://doi.org/10.1103/PRXQuantum.5.020332} {\bibfield  {journal} {\bibinfo  {journal} {PRX Quantum}\ }\textbf {\bibinfo {volume} {5}},\ \bibinfo {pages} {020332} (\bibinfo {year} {2024}{\natexlab{c}})}\BibitemShut {NoStop}%
\bibitem [{\citenamefont {Kasatkin}\ \emph {et~al.}(2023)\citenamefont {Kasatkin}, \citenamefont {Gu},\ and\ \citenamefont {Lidar}}]{kasatkin2023differential}%
  \BibitemOpen
  \bibfield  {author} {\bibinfo {author} {\bibfnamefont {V.}~\bibnamefont {Kasatkin}}, \bibinfo {author} {\bibfnamefont {L.}~\bibnamefont {Gu}},\ and\ \bibinfo {author} {\bibfnamefont {D.~A.}\ \bibnamefont {Lidar}},\ }\bibfield  {title} {\bibinfo {title} {Which differential equations correspond to the lindblad equation?},\ }\href@noop {} {\bibfield  {journal} {\bibinfo  {journal} {Physical Review Research}\ }\textbf {\bibinfo {volume} {5}},\ \bibinfo {pages} {043163} (\bibinfo {year} {2023})}\BibitemShut {NoStop}%
\bibitem [{\citenamefont {Nielsen}\ and\ \citenamefont {Chuang}(2010)}]{nielsen2010quantum}%
  \BibitemOpen
  \bibfield  {author} {\bibinfo {author} {\bibfnamefont {M.~A.}\ \bibnamefont {Nielsen}}\ and\ \bibinfo {author} {\bibfnamefont {I.~L.}\ \bibnamefont {Chuang}},\ }\href@noop {} {\emph {\bibinfo {title} {Quantum computation and quantum information}}}\ (\bibinfo  {publisher} {Cambridge university press},\ \bibinfo {year} {2010})\BibitemShut {NoStop}%
\bibitem [{\citenamefont {Shang}\ and\ \citenamefont {Zhao}(2024)}]{shang2024estimating}%
  \BibitemOpen
  \bibfield  {author} {\bibinfo {author} {\bibfnamefont {Z.-X.}\ \bibnamefont {Shang}}\ and\ \bibinfo {author} {\bibfnamefont {Q.}~\bibnamefont {Zhao}},\ }\bibfield  {title} {\bibinfo {title} {Estimating quantum amplitudes can be exponentially improved},\ }\href@noop {} {\bibfield  {journal} {\bibinfo  {journal} {arXiv preprint arXiv:2408.13721}\ } (\bibinfo {year} {2024})}\BibitemShut {NoStop}%
\bibitem [{sup()}]{supp}%
  \BibitemOpen
  \href@noop {} {}\bibinfo {note} {See Supplemental Material (Appendix) for details, which includes Ref [89-101].}\BibitemShut {Stop}%
\bibitem [{\citenamefont {An}\ \emph {et~al.}(2023{\natexlab{c}})\citenamefont {An}, \citenamefont {Liu}, \citenamefont {Wang},\ and\ \citenamefont {Zhao}}]{an2023theoryquantumdifferentialequation}%
  \BibitemOpen
  \bibfield  {author} {\bibinfo {author} {\bibfnamefont {D.}~\bibnamefont {An}}, \bibinfo {author} {\bibfnamefont {J.-P.}\ \bibnamefont {Liu}}, \bibinfo {author} {\bibfnamefont {D.}~\bibnamefont {Wang}},\ and\ \bibinfo {author} {\bibfnamefont {Q.}~\bibnamefont {Zhao}},\ }\href {https://arxiv.org/abs/2211.05246} {\bibinfo {title} {A theory of quantum differential equation solvers: limitations and fast-forwarding}} (\bibinfo {year} {2023}{\natexlab{c}}),\ \Eprint {https://arxiv.org/abs/2211.05246} {arXiv:2211.05246 [quant-ph]} \BibitemShut {NoStop}%
\bibitem [{\citenamefont {Goussev}\ \emph {et~al.}(2012)\citenamefont {Goussev}, \citenamefont {Jalabert}, \citenamefont {Pastawski},\ and\ \citenamefont {Wisniacki}}]{goussev2012loschmidt}%
  \BibitemOpen
  \bibfield  {author} {\bibinfo {author} {\bibfnamefont {A.}~\bibnamefont {Goussev}}, \bibinfo {author} {\bibfnamefont {R.~A.}\ \bibnamefont {Jalabert}}, \bibinfo {author} {\bibfnamefont {H.~M.}\ \bibnamefont {Pastawski}},\ and\ \bibinfo {author} {\bibfnamefont {D.}~\bibnamefont {Wisniacki}},\ }\bibfield  {title} {\bibinfo {title} {Loschmidt echo},\ }\href@noop {} {\bibfield  {journal} {\bibinfo  {journal} {arXiv preprint arXiv:1206.6348}\ } (\bibinfo {year} {2012})}\BibitemShut {NoStop}%
\bibitem [{\citenamefont {Brassard}\ \emph {et~al.}(2002{\natexlab{a}})\citenamefont {Brassard}, \citenamefont {Hoyer}, \citenamefont {Mosca},\ and\ \citenamefont {Tapp}}]{brassard2002quantum}%
  \BibitemOpen
  \bibfield  {author} {\bibinfo {author} {\bibfnamefont {G.}~\bibnamefont {Brassard}}, \bibinfo {author} {\bibfnamefont {P.}~\bibnamefont {Hoyer}}, \bibinfo {author} {\bibfnamefont {M.}~\bibnamefont {Mosca}},\ and\ \bibinfo {author} {\bibfnamefont {A.}~\bibnamefont {Tapp}},\ }\bibfield  {title} {\bibinfo {title} {Quantum amplitude amplification and estimation},\ }\href@noop {} {\bibfield  {journal} {\bibinfo  {journal} {Contemporary Mathematics}\ }\textbf {\bibinfo {volume} {305}},\ \bibinfo {pages} {53} (\bibinfo {year} {2002}{\natexlab{a}})}\BibitemShut {NoStop}%
\bibitem [{\citenamefont {Aaronson}\ and\ \citenamefont {Rall}()}]{aaronson2020quantum}%
  \BibitemOpen
  \bibfield  {author} {\bibinfo {author} {\bibfnamefont {S.}~\bibnamefont {Aaronson}}\ and\ \bibinfo {author} {\bibfnamefont {P.}~\bibnamefont {Rall}},\ }\bibinfo {title} {Quantum approximate counting, simplified},\ in\ \href {https://doi.org/10.1137/1.9781611976014.5} {\emph {\bibinfo {booktitle} {2020 Symposium on Simplicity in Algorithms (SOSA)}}},\ pp.\ \bibinfo {pages} {24--32}\BibitemShut {NoStop}%
\bibitem [{\citenamefont {Grinko}\ \emph {et~al.}(2021)\citenamefont {Grinko}, \citenamefont {Gacon}, \citenamefont {Zoufal},\ and\ \citenamefont {Woerner}}]{grinko2021iterative}%
  \BibitemOpen
  \bibfield  {author} {\bibinfo {author} {\bibfnamefont {D.}~\bibnamefont {Grinko}}, \bibinfo {author} {\bibfnamefont {J.}~\bibnamefont {Gacon}}, \bibinfo {author} {\bibfnamefont {C.}~\bibnamefont {Zoufal}},\ and\ \bibinfo {author} {\bibfnamefont {S.}~\bibnamefont {Woerner}},\ }\bibfield  {title} {\bibinfo {title} {Iterative quantum amplitude estimation},\ }\href {https://doi.org/10.1038/s41534-021-00379-1} {\bibfield  {journal} {\bibinfo  {journal} {npj Quantum Information}\ }\textbf {\bibinfo {volume} {7}},\ \bibinfo {pages} {52} (\bibinfo {year} {2021})}\BibitemShut {NoStop}%
\bibitem [{\citenamefont {Aharonov}\ \emph {et~al.}(2006)\citenamefont {Aharonov}, \citenamefont {Jones},\ and\ \citenamefont {Landau}}]{aharonov2006polynomial}%
  \BibitemOpen
  \bibfield  {author} {\bibinfo {author} {\bibfnamefont {D.}~\bibnamefont {Aharonov}}, \bibinfo {author} {\bibfnamefont {V.}~\bibnamefont {Jones}},\ and\ \bibinfo {author} {\bibfnamefont {Z.}~\bibnamefont {Landau}},\ }\bibfield  {title} {\bibinfo {title} {A polynomial quantum algorithm for approximating the jones polynomial},\ }in\ \href@noop {} {\emph {\bibinfo {booktitle} {Proceedings of the thirty-eighth annual ACM symposium on Theory of computing}}}\ (\bibinfo {year} {2006})\ pp.\ \bibinfo {pages} {427--436}\BibitemShut {NoStop}%
\bibitem [{\citenamefont {Brassard}\ \emph {et~al.}(2002{\natexlab{b}})\citenamefont {Brassard}, \citenamefont {Høyer}, \citenamefont {Mosca},\ and\ \citenamefont {Tapp}}]{Brassard_2002}%
  \BibitemOpen
  \bibfield  {author} {\bibinfo {author} {\bibfnamefont {G.}~\bibnamefont {Brassard}}, \bibinfo {author} {\bibfnamefont {P.}~\bibnamefont {Høyer}}, \bibinfo {author} {\bibfnamefont {M.}~\bibnamefont {Mosca}},\ and\ \bibinfo {author} {\bibfnamefont {A.}~\bibnamefont {Tapp}},\ }\href {https://doi.org/10.1090/conm/305/05215} {\bibinfo {title} {Quantum amplitude amplification and estimation}} (\bibinfo {year} {2002}{\natexlab{b}})\BibitemShut {NoStop}%
\bibitem [{\citenamefont {Grover}(1996)}]{grover1996fast}%
  \BibitemOpen
  \bibfield  {author} {\bibinfo {author} {\bibfnamefont {L.~K.}\ \bibnamefont {Grover}},\ }\href {https://arxiv.org/abs/quant-ph/9605043} {\bibinfo {title} {A fast quantum mechanical algorithm for database search}} (\bibinfo {year} {1996}),\ \Eprint {https://arxiv.org/abs/quant-ph/9605043} {arXiv:quant-ph/9605043 [quant-ph]} \BibitemShut {NoStop}%
\bibitem [{\citenamefont {Note1}()}]{commentqevt}%
  \BibitemOpen
  \bibfield  {author} {\bibinfo {author} {\bibnamefont {Note1}},\ }\href@noop {} {}\bibinfo {note} {We adapt the results described as Eq.~(318) in the original paper. The complexity of QEVT method actually depends on an additional factor $\alpha_{\widetilde{T},\psi}$ which is the upper bound on the maximum shifted partial sum of the Chebyshev expansion. For special cases where $V(t)$ is diagonalizable with all real eigenvalues, the complexity can be simplified to Eq.~(38) in the original paper with a quadratic dependence on the Jordan condition number.}\BibitemShut {Stop}%
\bibitem [{\citenamefont {Gilyén}\ and\ \citenamefont {Poremba}(2022)}]{gilyén2022improvedquantumalgorithmsfidelity}%
  \BibitemOpen
  \bibfield  {author} {\bibinfo {author} {\bibfnamefont {A.}~\bibnamefont {Gilyén}}\ and\ \bibinfo {author} {\bibfnamefont {A.}~\bibnamefont {Poremba}},\ }\href {https://arxiv.org/abs/2203.15993} {\bibinfo {title} {Improved quantum algorithms for fidelity estimation}} (\bibinfo {year} {2022}),\ \Eprint {https://arxiv.org/abs/2203.15993} {arXiv:2203.15993 [quant-ph]} \BibitemShut {NoStop}%
\bibitem [{\citenamefont {Note2}()}]{commentdelta}%
  \BibitemOpen
  \bibfield  {author} {\bibinfo {author} {\bibnamefont {Note2}},\ }\href@noop {} {}\bibinfo {note} {When $B=0$, the ODE is reduced to the Schrodinger equation. Since there are no non-zero eigenvalues, the $\Delta$ parameter will simply disappear. When $B$ is non-zero and the ground energy of $B$ is shifted to zero, then $\Delta$ is simply the gap of the ``Hamiltonian" $B$, and we will appreciate the gap to be polynomially small for complexity. When the ground energy of $B$ is above zero, $\Delta$ is simply the value of the ground energy.}\BibitemShut {Stop}%
\bibitem [{\citenamefont {Yuto~Ashida}\ and\ \citenamefont {Ueda}(2020)}]{ashida2020nonhermitian}%
  \BibitemOpen
  \bibfield  {author} {\bibinfo {author} {\bibfnamefont {Z.~G.}\ \bibnamefont {Yuto~Ashida}}\ and\ \bibinfo {author} {\bibfnamefont {M.}~\bibnamefont {Ueda}},\ }\bibfield  {title} {\bibinfo {title} {Non-hermitian physics},\ }\href {https://doi.org/10.1080/00018732.2021.1876991} {\bibfield  {journal} {\bibinfo  {journal} {Advances in Physics}\ }\textbf {\bibinfo {volume} {69}},\ \bibinfo {pages} {249} (\bibinfo {year} {2020})},\ \Eprint {https://arxiv.org/abs/https://doi.org/10.1080/00018732.2021.1876991} {https://doi.org/10.1080/00018732.2021.1876991} \BibitemShut {NoStop}%
\bibitem [{\citenamefont {Low}\ \emph {et~al.}(2025)\citenamefont {Low}, \citenamefont {King}, \citenamefont {Berry}, \citenamefont {Han}, \citenamefont {DePrince~III}, \citenamefont {White}, \citenamefont {Babbush}, \citenamefont {Somma},\ and\ \citenamefont {Rubin}}]{low2025fast}%
  \BibitemOpen
  \bibfield  {author} {\bibinfo {author} {\bibfnamefont {G.~H.}\ \bibnamefont {Low}}, \bibinfo {author} {\bibfnamefont {R.}~\bibnamefont {King}}, \bibinfo {author} {\bibfnamefont {D.~W.}\ \bibnamefont {Berry}}, \bibinfo {author} {\bibfnamefont {Q.}~\bibnamefont {Han}}, \bibinfo {author} {\bibfnamefont {A.~E.}\ \bibnamefont {DePrince~III}}, \bibinfo {author} {\bibfnamefont {A.}~\bibnamefont {White}}, \bibinfo {author} {\bibfnamefont {R.}~\bibnamefont {Babbush}}, \bibinfo {author} {\bibfnamefont {R.~D.}\ \bibnamefont {Somma}},\ and\ \bibinfo {author} {\bibfnamefont {N.~C.}\ \bibnamefont {Rubin}},\ }\bibfield  {title} {\bibinfo {title} {Fast quantum simulation of electronic structure by spectrum amplification},\ }\href@noop {} {\bibfield  {journal} {\bibinfo  {journal} {arXiv preprint arXiv:2502.15882}\ } (\bibinfo {year} {2025})}\BibitemShut {NoStop}%
\bibitem [{\citenamefont {King}\ \emph {et~al.}(2025)\citenamefont {King}, \citenamefont {Low}, \citenamefont {Babbush}, \citenamefont {Somma},\ and\ \citenamefont {Rubin}}]{king2025quantumsimulationsumofsquaresspectral}%
  \BibitemOpen
  \bibfield  {author} {\bibinfo {author} {\bibfnamefont {R.}~\bibnamefont {King}}, \bibinfo {author} {\bibfnamefont {G.~H.}\ \bibnamefont {Low}}, \bibinfo {author} {\bibfnamefont {R.}~\bibnamefont {Babbush}}, \bibinfo {author} {\bibfnamefont {R.~D.}\ \bibnamefont {Somma}},\ and\ \bibinfo {author} {\bibfnamefont {N.~C.}\ \bibnamefont {Rubin}},\ }\href {https://arxiv.org/abs/2505.01528} {\bibinfo {title} {Quantum simulation with sum-of-squares spectral amplification}} (\bibinfo {year} {2025}),\ \Eprint {https://arxiv.org/abs/2505.01528} {arXiv:2505.01528 [quant-ph]} \BibitemShut {NoStop}%
\bibitem [{\citenamefont {Low}\ and\ \citenamefont {Wiebe}(2018)}]{low2018hamiltonian}%
  \BibitemOpen
  \bibfield  {author} {\bibinfo {author} {\bibfnamefont {G.~H.}\ \bibnamefont {Low}}\ and\ \bibinfo {author} {\bibfnamefont {N.}~\bibnamefont {Wiebe}},\ }\bibfield  {title} {\bibinfo {title} {Hamiltonian simulation in the interaction picture},\ }\href@noop {} {\bibfield  {journal} {\bibinfo  {journal} {arXiv preprint arXiv:1805.00675}\ } (\bibinfo {year} {2018})}\BibitemShut {NoStop}%
\bibitem [{\citenamefont {Berry}\ \emph {et~al.}(2007)\citenamefont {Berry}, \citenamefont {Ahokas}, \citenamefont {Cleve},\ and\ \citenamefont {Sanders}}]{berry2007efficient}%
  \BibitemOpen
  \bibfield  {author} {\bibinfo {author} {\bibfnamefont {D.~W.}\ \bibnamefont {Berry}}, \bibinfo {author} {\bibfnamefont {G.}~\bibnamefont {Ahokas}}, \bibinfo {author} {\bibfnamefont {R.}~\bibnamefont {Cleve}},\ and\ \bibinfo {author} {\bibfnamefont {B.~C.}\ \bibnamefont {Sanders}},\ }\bibfield  {title} {\bibinfo {title} {Efficient quantum algorithms for simulating sparse hamiltonians},\ }\href@noop {} {\bibfield  {journal} {\bibinfo  {journal} {Communications in Mathematical Physics}\ }\textbf {\bibinfo {volume} {270}},\ \bibinfo {pages} {359} (\bibinfo {year} {2007})}\BibitemShut {NoStop}%
\bibitem [{\citenamefont {Berry}\ \emph {et~al.}(2014)\citenamefont {Berry}, \citenamefont {Childs}, \citenamefont {Cleve}, \citenamefont {Kothari},\ and\ \citenamefont {Somma}}]{berry2014exponential}%
  \BibitemOpen
  \bibfield  {author} {\bibinfo {author} {\bibfnamefont {D.~W.}\ \bibnamefont {Berry}}, \bibinfo {author} {\bibfnamefont {A.~M.}\ \bibnamefont {Childs}}, \bibinfo {author} {\bibfnamefont {R.}~\bibnamefont {Cleve}}, \bibinfo {author} {\bibfnamefont {R.}~\bibnamefont {Kothari}},\ and\ \bibinfo {author} {\bibfnamefont {R.~D.}\ \bibnamefont {Somma}},\ }\bibfield  {title} {\bibinfo {title} {Exponential improvement in precision for simulating sparse hamiltonians},\ }in\ \href@noop {} {\emph {\bibinfo {booktitle} {Proceedings of the forty-sixth annual ACM symposium on Theory of computing}}}\ (\bibinfo {year} {2014})\ pp.\ \bibinfo {pages} {283--292}\BibitemShut {NoStop}%
\bibitem [{\citenamefont {Braun}\ and\ \citenamefont {Golubitsky}(1983)}]{braun1983differential}%
  \BibitemOpen
  \bibfield  {author} {\bibinfo {author} {\bibfnamefont {M.}~\bibnamefont {Braun}}\ and\ \bibinfo {author} {\bibfnamefont {M.}~\bibnamefont {Golubitsky}},\ }\href@noop {} {\emph {\bibinfo {title} {Differential equations and their applications}}},\ Vol.~\bibinfo {volume} {2}\ (\bibinfo  {publisher} {Springer},\ \bibinfo {year} {1983})\BibitemShut {NoStop}%
\bibitem [{\citenamefont {Roccati}\ \emph {et~al.}(2022)\citenamefont {Roccati}, \citenamefont {Palma}, \citenamefont {Ciccarello},\ and\ \citenamefont {Bagarello}}]{roccati2022non}%
  \BibitemOpen
  \bibfield  {author} {\bibinfo {author} {\bibfnamefont {F.}~\bibnamefont {Roccati}}, \bibinfo {author} {\bibfnamefont {G.~M.}\ \bibnamefont {Palma}}, \bibinfo {author} {\bibfnamefont {F.}~\bibnamefont {Ciccarello}},\ and\ \bibinfo {author} {\bibfnamefont {F.}~\bibnamefont {Bagarello}},\ }\bibfield  {title} {\bibinfo {title} {Non-hermitian physics and master equations},\ }\href@noop {} {\bibfield  {journal} {\bibinfo  {journal} {Open Systems \& Information Dynamics}\ }\textbf {\bibinfo {volume} {29}},\ \bibinfo {pages} {2250004} (\bibinfo {year} {2022})}\BibitemShut {NoStop}%
\bibitem [{\citenamefont {Chakraborty}\ \emph {et~al.}(2019)\citenamefont {Chakraborty}, \citenamefont {Gily\'{e}n},\ and\ \citenamefont {Jeffery}}]{chakraborty2019power}%
  \BibitemOpen
  \bibfield  {author} {\bibinfo {author} {\bibfnamefont {S.}~\bibnamefont {Chakraborty}}, \bibinfo {author} {\bibfnamefont {A.}~\bibnamefont {Gily\'{e}n}},\ and\ \bibinfo {author} {\bibfnamefont {S.}~\bibnamefont {Jeffery}},\ }\bibfield  {title} {\bibinfo {title} {{The Power of Block-Encoded Matrix Powers: Improved Regression Techniques via Faster Hamiltonian Simulation}},\ }in\ \href {https://doi.org/10.4230/LIPIcs.ICALP.2019.33} {\emph {\bibinfo {booktitle} {46th International Colloquium on Automata, Languages, and Programming (ICALP 2019)}}},\ \bibinfo {series} {Leibniz International Proceedings in Informatics (LIPIcs)}, Vol.\ \bibinfo {volume} {132},\ \bibinfo {editor} {edited by\ \bibinfo {editor} {\bibfnamefont {C.}~\bibnamefont {Baier}}, \bibinfo {editor} {\bibfnamefont {I.}~\bibnamefont {Chatzigiannakis}}, \bibinfo {editor} {\bibfnamefont {P.}~\bibnamefont {Flocchini}},\ and\ \bibinfo {editor} {\bibfnamefont {S.}~\bibnamefont {Leonardi}}}\ (\bibinfo  {publisher} {Schloss Dagstuhl -- Leibniz-Zentrum
  f{\"u}r Informatik},\ \bibinfo {address} {Dagstuhl, Germany},\ \bibinfo {year} {2019})\ pp.\ \bibinfo {pages} {33:1--33:14}\BibitemShut {NoStop}%
\bibitem [{\citenamefont {Gilyén}(2019)}]{gilyenthesis}%
  \BibitemOpen
  \bibfield  {author} {\bibinfo {author} {\bibfnamefont {A.}~\bibnamefont {Gilyén}},\ }\emph {\bibinfo {title} {Quantum singular value transformation and its algorithmic applications}},\ \href@noop {} {Ph.D. thesis} (\bibinfo {year} {2019})\BibitemShut {NoStop}%
\bibitem [{\citenamefont {Note3}()}]{commentancilla}%
  \BibitemOpen
  \bibfield  {author} {\bibinfo {author} {\bibnamefont {Note3}},\ }\href@noop {} {}\bibinfo {note} {The number of ancillas for the Lindbladian simulation algorithms in arxiv: 2212.02051 and arxiv: 2405.19245 is of the order Eq. 77 in arxiv: 2212.02051. Since these ancilla can be reused, the number won't increase linearly in time.}\BibitemShut {Stop}%
\bibitem [{\citenamefont {Childs}\ \emph {et~al.}(2017{\natexlab{b}})\citenamefont {Childs}, \citenamefont {Kothari},\ and\ \citenamefont {Somma}}]{childs2017quantum}%
  \BibitemOpen
  \bibfield  {author} {\bibinfo {author} {\bibfnamefont {A.~M.}\ \bibnamefont {Childs}}, \bibinfo {author} {\bibfnamefont {R.}~\bibnamefont {Kothari}},\ and\ \bibinfo {author} {\bibfnamefont {R.~D.}\ \bibnamefont {Somma}},\ }\bibfield  {title} {\bibinfo {title} {Quantum algorithm for systems of linear equations with exponentially improved dependence on precision},\ }\href@noop {} {\bibfield  {journal} {\bibinfo  {journal} {SIAM Journal on Computing}\ }\textbf {\bibinfo {volume} {46}},\ \bibinfo {pages} {1920} (\bibinfo {year} {2017}{\natexlab{b}})}\BibitemShut {NoStop}%
\bibitem [{\citenamefont {Albert}\ and\ \citenamefont {Jiang}(2014)}]{albert2014symmetries}%
  \BibitemOpen
  \bibfield  {author} {\bibinfo {author} {\bibfnamefont {V.~V.}\ \bibnamefont {Albert}}\ and\ \bibinfo {author} {\bibfnamefont {L.}~\bibnamefont {Jiang}},\ }\bibfield  {title} {\bibinfo {title} {Symmetries and conserved quantities in lindblad master equations},\ }\href@noop {} {\bibfield  {journal} {\bibinfo  {journal} {Physical Review A}\ }\textbf {\bibinfo {volume} {89}},\ \bibinfo {pages} {022118} (\bibinfo {year} {2014})}\BibitemShut {NoStop}%
\bibitem [{\citenamefont {Bennett}\ \emph {et~al.}(1997)\citenamefont {Bennett}, \citenamefont {Bernstein}, \citenamefont {Brassard},\ and\ \citenamefont {Vazirani}}]{bennett1997strengths}%
  \BibitemOpen
  \bibfield  {author} {\bibinfo {author} {\bibfnamefont {C.~H.}\ \bibnamefont {Bennett}}, \bibinfo {author} {\bibfnamefont {E.}~\bibnamefont {Bernstein}}, \bibinfo {author} {\bibfnamefont {G.}~\bibnamefont {Brassard}},\ and\ \bibinfo {author} {\bibfnamefont {U.}~\bibnamefont {Vazirani}},\ }\bibfield  {title} {\bibinfo {title} {Strengths and weaknesses of quantum computing},\ }\href@noop {} {\bibfield  {journal} {\bibinfo  {journal} {SIAM journal on Computing}\ }\textbf {\bibinfo {volume} {26}},\ \bibinfo {pages} {1510} (\bibinfo {year} {1997})}\BibitemShut {NoStop}%
\bibitem [{\citenamefont {Abrams}\ and\ \citenamefont {Lloyd}(1998)}]{abrams1998nonlinear}%
  \BibitemOpen
  \bibfield  {author} {\bibinfo {author} {\bibfnamefont {D.~S.}\ \bibnamefont {Abrams}}\ and\ \bibinfo {author} {\bibfnamefont {S.}~\bibnamefont {Lloyd}},\ }\bibfield  {title} {\bibinfo {title} {Nonlinear quantum mechanics implies polynomial-time solution for np-complete and\# p problems},\ }\href@noop {} {\bibfield  {journal} {\bibinfo  {journal} {Physical Review Letters}\ }\textbf {\bibinfo {volume} {81}},\ \bibinfo {pages} {3992} (\bibinfo {year} {1998})}\BibitemShut {NoStop}%
\bibitem [{\citenamefont {Childs}\ and\ \citenamefont {Young}(2016)}]{childs2016optimal}%
  \BibitemOpen
  \bibfield  {author} {\bibinfo {author} {\bibfnamefont {A.~M.}\ \bibnamefont {Childs}}\ and\ \bibinfo {author} {\bibfnamefont {J.}~\bibnamefont {Young}},\ }\bibfield  {title} {\bibinfo {title} {Optimal state discrimination and unstructured search in nonlinear quantum mechanics},\ }\href@noop {} {\bibfield  {journal} {\bibinfo  {journal} {Physical Review A}\ }\textbf {\bibinfo {volume} {93}},\ \bibinfo {pages} {022314} (\bibinfo {year} {2016})}\BibitemShut {NoStop}%
\bibitem [{\citenamefont {Beals}\ \emph {et~al.}(2001)\citenamefont {Beals}, \citenamefont {Buhrman}, \citenamefont {Cleve}, \citenamefont {Mosca},\ and\ \citenamefont {De~Wolf}}]{beals2001quantum}%
  \BibitemOpen
  \bibfield  {author} {\bibinfo {author} {\bibfnamefont {R.}~\bibnamefont {Beals}}, \bibinfo {author} {\bibfnamefont {H.}~\bibnamefont {Buhrman}}, \bibinfo {author} {\bibfnamefont {R.}~\bibnamefont {Cleve}}, \bibinfo {author} {\bibfnamefont {M.}~\bibnamefont {Mosca}},\ and\ \bibinfo {author} {\bibfnamefont {R.}~\bibnamefont {De~Wolf}},\ }\bibfield  {title} {\bibinfo {title} {Quantum lower bounds by polynomials},\ }\href@noop {} {\bibfield  {journal} {\bibinfo  {journal} {Journal of the ACM (JACM)}\ }\textbf {\bibinfo {volume} {48}},\ \bibinfo {pages} {778} (\bibinfo {year} {2001})}\BibitemShut {NoStop}%
\bibitem [{\citenamefont {Song}\ \emph {et~al.}(2019)\citenamefont {Song}, \citenamefont {Yao},\ and\ \citenamefont {Wang}}]{song2019non}%
  \BibitemOpen
  \bibfield  {author} {\bibinfo {author} {\bibfnamefont {F.}~\bibnamefont {Song}}, \bibinfo {author} {\bibfnamefont {S.}~\bibnamefont {Yao}},\ and\ \bibinfo {author} {\bibfnamefont {Z.}~\bibnamefont {Wang}},\ }\bibfield  {title} {\bibinfo {title} {Non-hermitian skin effect and chiral damping in open quantum systems},\ }\href@noop {} {\bibfield  {journal} {\bibinfo  {journal} {Physical review letters}\ }\textbf {\bibinfo {volume} {123}},\ \bibinfo {pages} {170401} (\bibinfo {year} {2019})}\BibitemShut {NoStop}%
\bibitem [{\citenamefont {Minganti}\ \emph {et~al.}(2019)\citenamefont {Minganti}, \citenamefont {Miranowicz}, \citenamefont {Chhajlany},\ and\ \citenamefont {Nori}}]{minganti2019quantum}%
  \BibitemOpen
  \bibfield  {author} {\bibinfo {author} {\bibfnamefont {F.}~\bibnamefont {Minganti}}, \bibinfo {author} {\bibfnamefont {A.}~\bibnamefont {Miranowicz}}, \bibinfo {author} {\bibfnamefont {R.~W.}\ \bibnamefont {Chhajlany}},\ and\ \bibinfo {author} {\bibfnamefont {F.}~\bibnamefont {Nori}},\ }\bibfield  {title} {\bibinfo {title} {Quantum exceptional points of non-hermitian hamiltonians and liouvillians: The effects of quantum jumps},\ }\href@noop {} {\bibfield  {journal} {\bibinfo  {journal} {Physical Review A}\ }\textbf {\bibinfo {volume} {100}},\ \bibinfo {pages} {062131} (\bibinfo {year} {2019})}\BibitemShut {NoStop}%
\bibitem [{\citenamefont {Monkman}\ and\ \citenamefont {Berciu}(2024)}]{monkman2024limits}%
  \BibitemOpen
  \bibfield  {author} {\bibinfo {author} {\bibfnamefont {K.}~\bibnamefont {Monkman}}\ and\ \bibinfo {author} {\bibfnamefont {M.}~\bibnamefont {Berciu}},\ }\bibfield  {title} {\bibinfo {title} {Limits of the lindblad and non-hermitian description of open systems},\ }\href@noop {} {\bibfield  {journal} {\bibinfo  {journal} {arXiv preprint arXiv:2411.14599}\ } (\bibinfo {year} {2024})}\BibitemShut {NoStop}%
\bibitem [{\citenamefont {Jacobs}\ and\ \citenamefont {Steck}(2006)}]{jacobs2006straightforward}%
  \BibitemOpen
  \bibfield  {author} {\bibinfo {author} {\bibfnamefont {K.}~\bibnamefont {Jacobs}}\ and\ \bibinfo {author} {\bibfnamefont {D.~A.}\ \bibnamefont {Steck}},\ }\bibfield  {title} {\bibinfo {title} {A straightforward introduction to continuous quantum measurement},\ }\href@noop {} {\bibfield  {journal} {\bibinfo  {journal} {Contemporary Physics}\ }\textbf {\bibinfo {volume} {47}},\ \bibinfo {pages} {279} (\bibinfo {year} {2006})}\BibitemShut {NoStop}%
\end{thebibliography}%
\clearpage
\begin{appendix}
\onecolumngrid
\renewcommand{\addcontentsline}{\oldacl}
\renewcommand{\tocname}{Contents}
\textbf{Design nearly optimal quantum algorithm for linear differential equations via Lindbladians}

Zhong-Xia Shang, Naixu Guo, Dong An, and Qi Zhao


\tableofcontents





















\section{Preliminaries}
\subsection{Symbol conventions}
In this work, we will use $\|\cdot\|_2$ to denote the vector 2-norm for vectors, $\|\cdot\|$ to denote the spectral norm for matrices, and $\|\cdot\|_1$ to denote the trace norm for matrices.

In this work, we will use the symbol $n$ to denote the $n$-qubit system of ODE and the symbol $e$ to denote the environment system. We may sometimes omit these symbols when there is no ambiguity.



\subsection{Quantum linear algebra \label{sec:qla}}

In this subsection, we summarize some important results of quantum linear algebra.

\begin{definition}[Block encoding~\cite{chakraborty2019power, gilyen2019quantum}\label{def.blockencoding}]
We say a unitary $U_M$ is an $(\alpha,a,\epsilon)$-block encoding of matrix $M\in \mathbb{C}^{2^n\times 2^n}$ if
\begin{align}
    \|M-\alpha (\bra{0^a}\otimes I_n)U_M(\ket{0^a}\otimes I_n)\|\leq \epsilon.
\end{align}
\end{definition}
\noindent We may sometimes only use $\alpha$-block encoding as a shorter notation for $(\alpha,a,0)$.

\begin{lemma}[Multiplication of block encodings \cite{chakraborty2019power, gilyen2019quantum}]\label{product.blockencoding}
    If $U_{M_1}$ is an $(\alpha,a,\Delta)$-block encoding of an $s$-qubit operator $M_1$, and $U_{M_2}$ is a $(\beta,b,\epsilon)$-block encoding of an $s$-qubit operator $M_2$, then $(I_b\otimes U_{M_1})(I_a\otimes U_{M_2})$ is an $(\alpha\beta,a+b,\alpha\epsilon+\beta\Delta)$-block encoding of $M_1M_2$.
\end{lemma}

\begin{lemma}[Polynomial of a block encoding \cite{gilyen2019quantum}\label{theorem.qsvt}]
    Let $\xi>0$.
    Given $U_A$ that is an $(\alpha,a,\epsilon)$-block encoding of a Hermitian matrix $A$, and a real $\ell$-degree function $f(x)$ with $|f(x)|\leq \frac{1}{2}$ for $x\in [-1,1]$, one can prepare a $(1,a+n+4,4\ell\sqrt{\epsilon/\alpha}+\xi)$-block encoding of $f(A/\alpha)$ by using $\mathcal{O}(\ell)$ queries to $U_A$ and $\mathcal{O}(\ell(a+1))$ one- and two-qubit quantum gates.
    The description of the quantum circuit can be computed classically in time $\mathcal{O}(\mathrm{poly}(\ell,\log(1/\xi)))$.
\end{lemma}

Therefore, given the block encoding access of a matrix $A$, one can implement many different functions based on polynomial approximation.
In the following, we list and show some results about how to efficiently approximate important functions by polynomial approximation.

\begin{lemma}[Polynomial approximations of positive power functions \cite{gilyenthesis, gilyén2022improvedquantumalgorithmsfidelity}]\label{poly.positivepower}
    Let $\Delta,\epsilon\in (0,\frac{1}{2}]$, $c\in (0,1]$ and let $f(x):=\frac{1}{2}x^c$, then there exist real odd/even polynomials $P,P'$ such that for $x\in [\Delta,1]$, $|P(x)-f(x)|\leq \epsilon$, $|P'(x)-f(x)|\leq \epsilon$. Further, we have for $x\in [-1,1]$, $|P(x)|, |P'(x)|\leq 1$.
    The degrees of polynomials $P,P'$ are $\mathcal{O}(\frac{1}{\Delta}\log(\frac{1}{\epsilon}))$.
\end{lemma}

\begin{lemma}[Polynomial approximations of square root function\label{poly.sqrt}]
    Let $\Delta,\epsilon\in (0,\frac{1}{2}]$, and let $f(x)=\frac{1}{2}\sqrt{x}$. There exists a real odd polynomial $P$ with degree $\mathcal{O}(\frac{1}{\Delta}\log(\frac{1}{\epsilon}))$ such that for $x\in [\Delta,1]$, $|P(x)-f(x)|\leq \epsilon$, and $|P(x)|\leq 1$ for $x\in [-1,1]$.
\end{lemma}
\begin{proof}
    This can be achieved by setting $c=\frac{1}{2}$ and choosing the odd polynomial in Lemma~\ref{poly.positivepower}.
\end{proof}

\subsection{Quantum algorithms for open system simulation}

We first introduce the current state-of-the-art result for time-independent Lindbladian simulation \cite{li2023opensystem}. The basic idea in Ref \cite{li2023opensystem} is to use high-order series
expansion following the spirit of Duhamel’s principle \cite{hartman2002ordinary} where the authors separate the Lindbladian into two parts: the jumping part and the drifting part. The truncated expansion is then realized on quantum computers with the aid of linear combinations of unitaries (LCU) \cite{childs2012hamiltonians} and oblivious amplitude amplification for isometry \cite{cleve2016efficient}.

\begin{lemma}[Quantum algorithm for time-independent Lindbladian \cite{li2023opensystem}]\label{thm.opensystem.independent}
Given access to an $(\alpha_0, a, \epsilon')$-block encoding $U_H$ of $H$, and an $(\alpha_j, a, \epsilon')$-block encoding $U_{F_j}$ for each $F_j$, where $j\in[m]$.
Let $\alpha_{\mathcal{L}}:=\alpha_0+\frac{1}{2}\sum_{j=1}^m \alpha_j^2$.
For all $t, \epsilon \geq 0$ with $\epsilon' \leq \epsilon / (\alpha_{\mathcal{L}}T)$, there exists a quantum algorithm that outputs a purification of $\tilde{\rho}_t$ such that $\|\tilde{\rho}_t-e^{\mathcal{L}t}(\rho)\|_1\leq \epsilon$, using
\begin{equation}
    \mathcal{O} \left( \alpha_{\mathcal{L}} T \frac{\log(\alpha_{\mathcal{L}} T / \epsilon)}{\log \log(\alpha_{\mathcal{L}} T / \epsilon)} \right)
\end{equation}
queries to $U_H$ and $U_{F_j}$, $\mathcal{O} \left( m\alpha_{\mathcal{L}} T \text{poly}\log(\alpha_\mathcal{L}T\epsilon^{-1}) \right)$ additional 1- and 2-qubit gates, and $\mathcal{O}(\text{poly}\log(\alpha_{\mathcal{L}}T\epsilon^{-1}))$ number of additional ancilla qubits besides those used for block encodings \cite{commentancilla}. 
\end{lemma}


For the time-dependent case \cite{he2024efficientoptimalcontrolopen}, we consider the following input model. The basic idea in Ref \cite{he2024efficientoptimalcontrolopen} follows the time-independent case. Assume we are given access to an \((\alpha_0(t), a, \epsilon')\)-block encoding \(U_{H(t)}\) of \(H(t)\), and an \((\alpha_j(t), a, \epsilon')\)-block encoding \(U_{F_j(t)}\) for each \(F_j(t)\) for all \(0 \leq t \leq T\), i.e.,
\begin{align}
    (\bra{0}_a\otimes I)U_{H(t)}(\ket{0}_a\otimes I)\approx\sum_t |t\rangle \langle t|\otimes \frac{H(t)}{\alpha_0(t)},\\
    (\bra{0}_a\otimes I)U_{L_j(t)}(\ket{0}_a\otimes I)\approx\sum_t |t\rangle \langle t|\otimes \frac{L_j(t)}{\alpha_j(t)}.
\end{align}
Define
\begin{align}
    \alpha_{\mathcal{L'}}(T):= \int_0^T d\tau \alpha_{\mathcal{L}}(\tau),
\end{align}
where $\alpha_{\mathcal{L}}(\tau)=\alpha_0(\tau)+\frac{1}{2}\sum_{j=1}^m \alpha_j^2(\tau)$.
Assume we have access to
\begin{align}
    O_{\alpha_{\mathcal{L}'}}\ket{t}\ket{z}&=\ket{t}\ket{z\oplus \alpha^{-1}_{\mathcal{L}'}(t)},\\
    O_{H, \mathrm{norm}}\ket{t}\ket{z}&=\ket{t}\ket{z\oplus \alpha_0(t)},\\
    O_{F_j, \mathrm{norm}}\ket{t}\ket{z}&=\ket{t}\ket{z\oplus \alpha_j(t)}.
\end{align}

\begin{lemma}[Quantum algorithm for time-dependent Lindbladian \cite{he2024efficientoptimalcontrolopen}\label{thm.opensystem.dependent}]
Given access to an \((\alpha_0(t), a, \epsilon')\)-block encoding \(U_{H(t)}\) of \(H(t)\), and an \((\alpha_j(t), a, \epsilon')\)-block encoding \(U_{F_j(t)}\) for each \(F_j(t)\) for all \(0 \leq t \leq T\), where $j\in [m]$.
Suppose further that \(\epsilon' \leq \epsilon / (2T(m+1))\). Then, there exists a quantum algorithm that outputs a purification of \(\tilde{\rho}_T\) of \(\rho(T)\) where
$ \left\| \tilde{\rho}_T - \mathcal{T} e^{\int_0^T d\tau \mathcal{L}(\tau)} (\rho_0) \right\|_1 \leq \epsilon $ using 
\begin{equation}
\mathcal{O}\left(\alpha_{\mathcal{L'}}T \left( \frac{\log\left( \alpha_{\mathcal{L'}}(T) / \epsilon \right)}{\log \log \left( \alpha_{\mathcal{L'}}(T) / \epsilon \right)}\right)^2 \right)
\end{equation}
queries to \(U_{H(t)}\), \(U_{F_j(t)}, O_{\alpha_{\mathcal{L}'}}, O_{H, \mathrm{norm}}, O_{F_j, \mathrm{norm}}\), $\mathcal{O} \left( (m+n)\alpha_{\mathcal{L}'} T \text{poly}\log((\alpha_{\mathcal{L}'}+\beta_\mathcal{L})T\epsilon^{-1}) \right)$ additional 1- and 2-qubit gates, and $\mathcal{O}(\text{poly}\log((\alpha_{\mathcal{L}'} + \beta_{\mathcal{L}} ) T/\epsilon))$ number of additional ancilla qubits besides those used for block encodings, where we define $\beta_{\mathcal{L}} = \max_{t \in [0,T]} \|\frac{\partial H(t)}{\partial t}\| + \sum_{j=1}^m \max_{t\in[0,T]} \|\frac{\partial L_j(t)}{\partial t}\| $. 
\end{lemma}

One can also consider a simplified input model.
Let $\alpha_{\mathcal{L}}\geq \max_t \{\|H(t)\|, F_j(t)\}$ for $0\leq t\leq T$ and $j\in [m]$.
Note that we have $\alpha_{\mathcal{L}} T\geq \alpha_{\mathcal{L}'}(T)$.
Assume we are given access to an \((\alpha_{\mathcal{L}}, a, \epsilon')\)-block encoding \(U_{H(t)}\) of \(H(t)\), and an \((\alpha_{\mathcal{L}}, a, \epsilon')\)-block encoding \(U_{F_j(t)}\) for each \(F_j(t)\) for all \(0 \leq t \leq T\), i.e.,
\begin{align}
    (\bra{0}_a\otimes I)U_{H(t)}(\ket{0}_a\otimes I)\approx\sum_t |t\rangle \langle t|\otimes \frac{H(t)}{\alpha_{\mathcal{L}}},\\
    (\bra{0}_a\otimes I)U_{F_j(t)}(\ket{0}_a\otimes I)\approx\sum_t |t\rangle \langle t|\otimes \frac{F_j(t)}{\sqrt{\alpha_{\mathcal{L}}}}.
\end{align}
Since the embedding factor $\alpha_{\mathcal{L}}$ is time-independent in this case, if we know it classically, one can construct the oracles required in the previous input model.
Based on these modifications, we have the following lemma:

\begin{lemma}[Quantum algorithm for time-dependent Lindbladian, simplified\label{thm.opensystem.dependent.simple}]
Given access to an \((\alpha_{\mathcal{L}}, a, \epsilon')\)-block encoding \(U_{H(t)}\) of \(H(t)\), and an \((\alpha_{\mathcal{L}}, a, \epsilon')\)-block encoding \(U_{F_j(t)}\) for each \(F_j(t)\) for all \(0 \leq t \leq T\), where $j\in [m]$.
Let $\alpha_{\mathcal{L}}$ be known classically.
Suppose further that \(\epsilon' \leq \epsilon / (2T(m+1))\). Then, there exists a quantum algorithm that outputs a purification of \(\tilde{\rho}_T\) of \(\rho(T)\) where
$ \left\| \tilde{\rho}_T - \mathcal{T} e^{\int_0^T d\tau \mathcal{L}(\tau)} (\rho_0) \right\|_1 \leq \epsilon $
using
\begin{equation}
\mathcal{O}\left(\alpha_{\mathcal{L}} T \left( \frac{\log\left( \alpha_{\mathcal{L}}(T) / \epsilon \right)}{\log \log \left( \alpha_{\mathcal{L}}(T) / \epsilon \right)}\right)^2 \right)
\end{equation}
queries to \(U_{H(t)}\), \(U_{F_j(t)}\), $\mathcal{O} \left( (m+n)\alpha_{\mathcal{L}} T \text{poly}\log((\alpha_{\mathcal{L}}+\beta_\mathcal{L})T\epsilon^{-1}) \right)$ additional 1- and 2-qubit gates, and $\mathcal{O}(\text{poly}\log((\alpha_{\mathcal{L}} + \beta_{\mathcal{L}} ) T/\epsilon))$ number of additional ancilla qubits besides those used for block encodings.
\end{lemma}

\subsection{Amplitude amplification and amplitude estimation \label{sec:aaae}}
\begin{lemma}[Amplitude amplification with ancilla \cite{brassard2002quantum,childs2017quantum}]\label{aa}
Given an $n$-qubit oracle $U_i$ and a $(1+n)$-qubit oracle $U_m$ satisfy $U_i|0\rangle_n=|\psi_0\rangle$ and $U_m|0\rangle|\psi_0\rangle=c_0|0\rangle|\psi_1\rangle+c_1|1\rangle|\psi_2\rangle$, we can use: $$U_1=I_{1+n}-2|0\rangle\langle 0|\otimes I_n,$$and: $$U_2=U_m(I\otimes U_i(I_{1+n}-2|0\rangle\langle 0|_{1+n})I\otimes U_i^\dag)U_m^\dag,$$ to build the Grover operator $U_{Gr}=U_2 U_1$ \rm{\cite{grover1996fast}}. Implementing $U_{Gr}$ repeatedly for $\mathcal{O}(c_0^{-1})$ times on $U_m|0\rangle|\psi_0\rangle$, we are able to prepare a state with $\mathcal{O}(1)$ overlap with the state $|0\rangle|\psi_1\rangle$. 
\end{lemma}

\begin{lemma}[Improved amplitude estimation \cite{aaronson2020quantum,grinko2021iterative}]\label{ae}
Given two quantum states $|S_1\rangle=U_1|0\rangle$ and $|S_2\rangle=U_2|0\rangle$, we can query the Grover operator \rm{\cite{grover1996fast}}:
$$U_G=(I-2U_1|0\rangle\langle 0|U_1^\dag)(I-2U_2|0\rangle\langle 0|U_2^\dag),$$
for $\mathcal{O}(\varepsilon^{-1}\log(\delta^{-1}))$ times to estimate $\mu=|\langle S_2|S_1\rangle|$ up to an \textbf{additive} error $\varepsilon$ (i.e. $|\bar{\mu}-\mu|\leq \varepsilon$) with a success probability at least $1-\delta$.
\end{lemma}
Note that Lemma \ref{ae} is an improved version of the original amplitude estimation algorithm \cite{brassard2002quantum}. Compared with the original one, there is no need to build quantum Fourier transform circuits, and the dependency on $\delta$ has been improved from $\delta^{-1}$ to $\log(\delta^{-1})$.

\section{Solving linear ODEs via Lindbladians\label{sec: slovl}}
\subsection{Non-diagonal density matrix encoding}
Similar to the idea of block encoding where a matrix $M$ of interest is encoded into the upper-left block of a unitary operator, the non-diagonal density matrix encoding (NDME) \cite{shang2024estimating} aims to encode the matrix $M$ into a non-diagonal block of a density matrix:
\begin{definition}[Non-diagonal density matrix encoding (NDME)]
Given an $(l+n)$-qubit density matrix $\rho_M$ and a $n$-qubit matrix $M$, if $\rho_M$ satisfies:
$$(\langle s_1|_l\otimes I_n) \rho_M (|s_2\rangle_l\otimes I_n)=\gamma M,$$
where $|s_1\rangle_l$ and $|s_2\rangle_l$ are two different computational basis states of the $l$-qubit ancilla system and $\gamma\geq 0$, then $\rho_M$ is called an $(l+n,|s_1\rangle,|s_2\rangle,\gamma)$-NDME of $M$.
\end{definition}
\noindent The non-diagonal block of density matrices can let us jump out of the Hermitian and positive semi-definite restrictions of density matrices. In this work, we will focus on $(1+n,|0\rangle,|1\rangle,\gamma)$-NDME where $\rho_M$ has the form:
\begin{eqnarray}
\rho_M=\begin{pmatrix}
\cdot & \gamma M \\
\gamma M^\dag & \cdot
\end{pmatrix}.
\end{eqnarray}
Since $\rho_M$ is Hermitian, it is also a $(1+n,|1\rangle,|0\rangle,\gamma)$-NDME of $M^\dag$.

Given the definition of NDME, we can consider a class of operations we can implement on $M$. We consider applying a quantum channel $\mathcal{C}[\cdot]$ to $\rho_{M}$ with the form:
\begin{equation}
\mathcal{C}[\rho_{M}]=\sum_i \begin{pmatrix}
K_i & 0\\
0& L_i
\end{pmatrix}\rho_{M} \begin{pmatrix}
K_i^\dag & 0\\
0& L_i^\dag
\end{pmatrix}=\begin{pmatrix}
\cdot & \gamma\sum_i K_iML_i^\dag \\
\gamma\sum_i L_i M^\dag K_i^\dag & \cdot
\end{pmatrix},
\end{equation}
where $\{K_i\}$ and $\{L_i\}$ are two sets of Kraus operators with $\sum_i K_i^\dag K_i=\sum_i L_i^\dag L_i=I_n$, following the fact that quantum channels are completely positive trace-preserving maps.
Under the vectorization picture, it is equivalent that we realize the operation:
\begin{equation}
\vec{M}\rightarrow \left(\sum_i K_i\otimes L_i^* \right)\vec{M}.
\end{equation}
Thus, by adjusting $\{K_i\}$ and $\{L_i\}$, we can encode desired operations into $\sum_i K_i\otimes L_i^*$, which is known as the channel block encoding (CBE) (See the formal definition of CBE in Ref. \cite{shang2024estimating}).

We can consider a special case of NDME and CBE where $M=|\psi\rangle\langle\phi|$ and $L_i=l_i I_n$ (i.e. $\sum_i |l_i|^2=1$). This corresponds to $\rho_M=1/2(|0\rangle|\psi\rangle+|1\rangle|\phi\rangle)(\langle 0|\langle\psi|+\langle 1|\langle \phi|)$ with $\gamma=1/2$ which is transformed by the quantum channel into:
\begin{equation}\label{spcbe}
\frac{1}{2}\begin{pmatrix}
|\psi\rangle\langle\psi| & |\psi\rangle\langle\phi| \\
|\phi\rangle\langle\psi| & |\phi\rangle\langle\phi|
\end{pmatrix}\rightarrow \begin{pmatrix}
\sum_i K_i|\psi\rangle\langle\psi|K_i^\dag & \left(\sum_i l_i K_i|\psi\rangle\right)\langle\phi| \\
|\phi\rangle\left(\langle\psi|\sum_i l_i^* K_i^\dag \right)& |\phi\rangle\langle\phi|
\end{pmatrix}.
\end{equation}
Thus, on the upper-right block, we equivalently realize the operation $\sum_i l_i K_i$ on $|\psi\rangle$.

\subsection{Encoding ODEs into Lindbladians\label{sec.encode.odetolind}}
Lindbladian Eq. \ref{lmee} can be understood as a special ODE under the vectorization picture:
\begin{equation}\label{vlme}
\quad\frac{d\vec{\rho}}{dt}=L(t)\vec{\rho},
\end{equation}
where $\vec{\rho}=\sum_{ij}\rho_{ij}|i\rangle|j\rangle$ is the vector representation of the density matrix $\rho$ and $L$ is the Liouvillian generator for the Lindbladian semi-group which has the following matrix form \cite{albert2014symmetries,shang2024polynomial}:
\begin{eqnarray}\label{liou}
&&L(t)=-i\left(H(t)\otimes I-I\otimes H(t)^T\right)+\sum_i \left(F_i(t)\otimes F_i(t)^*-\frac{1}{2}F_i(t)^\dag F_i(t)\otimes I-I\otimes\frac{1}{2}F_i(t)^TF_i(t)^*\right).
\end{eqnarray}

Following the similar idea shown in Eq. \ref{spcbe}, we now show how to use Lindbladian to encode general linear ODEs. First, we set a $(1+n)$-qubit initial state $\rho_{0}$ which has the form:
\begin{eqnarray}
\rho_{0}=\frac{1}{2}\begin{pmatrix}
 |\mu_0\rangle\langle \mu_0| &  |\mu_0\rangle\langle \mu_0| \\
 |\mu_0\rangle\langle \mu_0| & |\mu_0\rangle\langle \mu_0|
\end{pmatrix}.
\end{eqnarray}
which is a $(1+n,|0\rangle,|1\rangle,\frac{1}{2})$-NDME of $|\mu_0\rangle\langle \mu_0|$. We then consider acting a Lindbladian of form Eq. \ref{lmee} with $H(t)=\begin{pmatrix}
H_1(t) & 0 \\
0 & \alpha I_n
\end{pmatrix}$ and $F_i(t)=\begin{pmatrix}
G_i(t) & 0 \\
0 & \beta_i I_n
\end{pmatrix}$ on $\rho_0$, then under the vectorization picture, the Liouvillian generator is of the form:
\begin{eqnarray}
L(t)&&=-i\begin{pmatrix}
H_1(t) & 0 \\
0 & \alpha I_n
\end{pmatrix}\otimes \begin{pmatrix}
I_n & 0 \\
0 & I_n
\end{pmatrix}+i \begin{pmatrix}
I_n & 0 \\
0 & I_n
\end{pmatrix}\otimes \begin{pmatrix}
H_1(t)^T & 0 \\
0 & \alpha I_n
\end{pmatrix}\nonumber\\&&+\sum_i \begin{pmatrix}
G_i(t) & 0 \\
0 & \beta_i I_n
\end{pmatrix}\otimes \begin{pmatrix}
G_i(t)^* & 0 \\
0 & \beta_i^* I_n
\end{pmatrix}\nonumber\\&&-\sum_i \frac{1}{2} \begin{pmatrix}
G_i(t)^\dag G_i(t) & 0 \\
0 & |\beta_i|^2 I_n
\end{pmatrix}\otimes \begin{pmatrix}
I & 0 \\
0 & I
\end{pmatrix}-\sum_i \frac{1}{2} \begin{pmatrix}
I & 0 \\
0 & I
\end{pmatrix}\otimes\begin{pmatrix}
G_i(t)^T G_i(t)^* & 0 \\
0 & |\beta_i|^2 I_n
\end{pmatrix}.
\end{eqnarray}

Focusing on the upper-right block of $\rho_0$, $\mathcal{L}[\rho_0]$ gives:
\begin{eqnarray}
&&\frac{1}{2}|\mu_0\rangle\langle \mu_0|\rightarrow \frac{1}{2}\left(-i(H_1(t)-\alpha)-\left(\frac{1}{2}\sum_i G_i(t)^\dag G_i(t) -\sum_i\beta_i^*G_i(t)+\frac{1}{2}\sum_i|\beta_i|^2\right)\right)|\mu_0\rangle\langle \mu_0|.
\end{eqnarray}
Here, since the $\alpha$ term can be absorbed into $H_1(t)$ and $\frac{1}{2}\sum_i G_i(t)^\dag G_i(t)$ is positive semi-definite, for semi-dissipative ODEs, we can simply set $\alpha=\beta_i=0$ such that on this upper-right block, we realize the homogenous linear ODE:
\begin{eqnarray}\label{dode}
\frac{d \vec{\mu}(t)}{dt}=-V(t)\vec{\mu}(t)=\left(-i H_1(t)-\frac{1}{2}\sum_i G_i(t)^\dag G_i(t)\right) \vec{\mu}(t),
\end{eqnarray}
with $\vec{\mu}(0)=|\mu_0\rangle$, $A(t)=\frac{V(t)-V^\dagger(t)}{2i}=H_1(t)$ and $B(t)=\frac{V(t)+V^\dagger(t)}{2}=\frac{1}{2}\sum_i G_i(t)^\dag G_i(t)$. 

Finally, at time $T$, we will have:
\begin{eqnarray}\label{rhot}
\rho_{T}=\frac{1}{2}\begin{pmatrix}
\sigma_T & \eta_T|\mu_T\rangle\langle \mu_0| \\
\eta_T|\mu_0\rangle\langle \mu_T| & |\mu_0\rangle\langle \mu_0|
\end{pmatrix}.
\end{eqnarray}
where $|\mu_T\rangle=\vec{\mu}(T)/\eta_T$ is the normalized solution with $\eta_T=\|\vec{\mu}(T)\|$ and $\sigma_T$ is the output density matrix from the initial state $|\mu_0\rangle\langle \mu_0|$ under the evolution of Lindbladian with $H_1(t)$ as the internal Hamiltonian and $G_i(t)$ as jump operators after a time $T$. 


\subsection{Measuring properties \label{sec:mp}}
In this subsection, without preparing $|\mu_T\rangle$, we consider how to measure properties of $\vec{\mu}(T)$ such as $\eta_T\langle \phi_0|\mu_T\rangle$ known as the Loschmidt echo \cite{goussev2012loschmidt} when $|\phi_0\rangle=|\mu_0\rangle$ and the expectation value $\eta_T^2\langle \mu_T|O|\mu_T\rangle$ with respect to an observable $O$. 

For $\eta_T\langle \phi_0|\mu_T\rangle$, we can follow the procedure in subsection \ref{sec.encode.odetolind} except for changing the initial state $\rho_0$ from $|+\rangle|\mu_0\rangle$ to $1/\sqrt{2}(|0\rangle|\mu_0\rangle+|1\rangle|\phi_0\rangle)$. Then, at time $T$, we will have $\rho_{T1}$ with the form:
\begin{eqnarray}\label{rhot1}
\rho_{T1}=\frac{1}{2}\begin{pmatrix}
\sigma_T & \eta_T|\mu_T\rangle\langle \phi_0| \\
\eta_T|\phi_0\rangle\langle \mu_T| & |\phi_0\rangle\langle \phi_0|
\end{pmatrix}.
\end{eqnarray}
From $\rho_{T1}$, we notice that:
\begin{eqnarray}
&&\text{Tr}((X\otimes I_n) \rho_{T1})=\frac{1}{2} \eta_T (\langle\phi_0|\mu_T\rangle+\langle\mu_T|\phi_0\rangle)=\eta_T \text{Re}[\langle \phi_0|\mu_T\rangle],\nonumber\\&&\text{Tr}((Y\otimes I_n)  \rho_{T1})=\frac{1}{2} \eta_T (-i\langle\phi_0|\mu_T\rangle+i\langle\mu_T|\phi_0\rangle)=\eta_T \text{Im}[\langle \phi_0|\mu_T\rangle]. 
\end{eqnarray}
If we define $|S_{T1}\rangle$ as the purification state of $\rho_{T1}$, then we have:
\begin{eqnarray}
\langle S_{T1}|I_e\otimes X\otimes I_n |S_{T1}\rangle&&=\eta_T \text{Re}[\langle \phi_0|\mu_T\rangle],\nonumber\\\langle S_{T1}|I_e\otimes Y\otimes I_n |S_{T1}\rangle&&=\eta_T \text{Im}[\langle \phi_0|\mu_T\rangle],
\end{eqnarray}
with $e$ denotes the environment system of $|S_{T1}\rangle$. Note that the $|S_{T1}\rangle$ can be directly obtained from the quantum algorithm of simulating Lindbladians \cite{li2023opensystem,he2024efficientoptimalcontrolopen}. Since we have now turned the values of interest into amplitudes, based on Lemma \ref{ae}, we can now use $\mathcal{O}(\eta_T^{-1}\varepsilon^{-1})$ queries on the preparation oracle of $|S_{T1}\rangle$ to give a $\eta_T\varepsilon$-additive estimation of $\eta_T\langle \phi_0|\mu_T\rangle$.

For $\eta_T^2\langle \mu_T|O|\mu_T\rangle$, we need to prepare another density matrix $\rho_{T2}$. Starting from $\rho_{T}$, we can further act a Lindbladian on it with $H'(t)=\begin{pmatrix}
0 & 0 \\
0 & H_1(t)
\end{pmatrix}$ and $F_i'(t)=\begin{pmatrix}
0 & 0 \\
0 & G_i(t)
\end{pmatrix}$. Then, following the similar calculation as in Section~\ref{sec.encode.odetolind}, after another time $T$, we will have:
\begin{eqnarray}\label{rhot2}
\rho_{T2}=\frac{1}{2}\begin{pmatrix}
\sigma_T &  \eta_T^2|\mu_T\rangle\langle \mu_T| \\
\eta_T^2|\mu_T\rangle\langle \mu_T| & \sigma_T. 
\end{pmatrix}.
\end{eqnarray}
From $\rho_{T2}$, we can observe that:
\begin{eqnarray}
\text{Tr}((X\otimes O)\rho_{T2})=\eta_T^2\langle \mu_T|O|\mu_T\rangle
\end{eqnarray}
For simplicity, we assume $O$ is a Pauli operator. Then, by replacing $\rho_{T2}$ to its purification state $|S_{T2}\rangle$, the value of interest is encoded as an amplitude. Based on Lemma \ref{ae}, if we want to have a $\eta_T^2\varepsilon$-additive estimation of $\eta_T^2\langle \mu_T|O|\mu_T\rangle$, we need $\mathcal{O}(\eta_T^{-2}\varepsilon^{-1})$ queries on the preparation oracle of $|S_{T2}\rangle$.

\subsection{Extracting $|\mu_T\rangle$ \label{sec: extract}}
In this subsection, we show how to prepare $|\mu_T\rangle$ from $\rho_T$. We have two oracles. First, we have the initial state preparation oracle: $U_I|0\rangle_n=|\mu_0\rangle$, then we have the Lindbladian simulation oracle $U_L$: $U_L(H\otimes U_I\otimes I_e)|0\rangle|0\rangle_n|0\rangle_e=|S_T\rangle$ which is the purification of $\rho_T$. The construction of the Lindbladian in subsection \ref{sec.encode.odetolind} makes $|S_T\rangle$ have the form:
\begin{eqnarray}
|S_T\rangle=\frac{1}{\sqrt{2}}(|0\rangle |S_{T0}\rangle_{n+e}+|1\rangle |\mu_0\rangle|E_{T1}\rangle_e)
\end{eqnarray}
where $|E_{T1}\rangle$ is an environment state. Compared with $\rho_T$, we must have:
\begin{eqnarray}
\text{Tr}_e(|S_{T0}\rangle_{n+e} \langle \mu_0|\langle E_{T1}|_e)=\eta_T |\mu_T\rangle\langle\mu_0|.
\end{eqnarray}
Thus, we know:
\begin{eqnarray}
( I_n\otimes \langle E_{T1}|_e)|S_{T0}\rangle=\eta_T |\mu_T\rangle.
\end{eqnarray}
We can re-express $|S_T\rangle$ into:
\begin{eqnarray}
|S_T\rangle=\frac{\eta_T}{\sqrt{2}}|0\rangle|\mu_T\rangle|E_{T1}\rangle+\sqrt{1-\eta_T^2/2}|junk\rangle_{1+n+e},
\end{eqnarray}
satisfying $(\langle 0|\otimes I_n\otimes \langle E_{T1}|)|junk\rangle_{1+n+e}=0$.

Since $U_L$ satisfies:
\begin{eqnarray}
U_L|1\rangle|\psi\rangle_n|0\rangle_e=|1\rangle|\psi\rangle_n|E_{T1}\rangle_e,
\end{eqnarray}
for any $|\psi\rangle_n$, we can build a reflection operator $U_1$:
\begin{eqnarray}
U_1&&=(X\otimes I_{n+e})U_L(I_{1+n+e}-2X|0\rangle\langle 0|X\otimes I_n\otimes |0\rangle\langle 0|_e)U_L^\dag (X\otimes I_{n+e})\nonumber\\=&&I_{1+n+e}-2|0\rangle\langle 0|\otimes I_n\otimes |E_{T1}\rangle\langle E_{T1}|_e.
\end{eqnarray}
We can also build another reflection operator $U_2$:
\begin{eqnarray}
U_2&&=U_L(H\otimes U_I\otimes I_e)(I_{1+n+e}-2|0\rangle\langle 0|_{1+n+e})(H\otimes U_I^\dag\otimes I_e)U_L^\dag.
\end{eqnarray}
The composite operator $U_2U_1$ which uses $2$ calls on $U_L$ and $U_L^\dag$ and uses a single call on $U_I$ and $U_I^\dag$, then serves as the Grover operator used in amplitude amplification. Based on Lemma \ref{aa}, we have the following lemma:
\begin{lemma}[Target extraction\label{lemma.state.extract}]
$|\mu_T\rangle$ can be prepared using $\mathcal{O}(\eta_T^{-1})$ queries on the Lindbladian simulation oracle $U_L$ and the initial state preparation oracle $U_I$.
\end{lemma}

\subsection{Algorithm and complexity\label{sec: algcom}}
\subsubsection{Time-independent case}
In this work, we consider two different input models: the first one called the direct access model is that we are directly given the block encoding of matrix $V$, and the other one called the square root access model is that we are given $1$-block encoding $U_{H_1,G_i}$ of $\alpha_V^{-1}|0\rangle\langle 0|\otimes H_1+\sum_i \alpha_V^{-1/2}|i\rangle\langle i|\otimes G_i$ with $H_1=\frac{V-V^\dagger}{2i}$, $\sum_{i=1}^D G_i^\dagger G_i=\frac{V+V^\dagger}{2}$, where $D$ is a constant. 

We start with the first input model.
For the Lindbladian simulation, we need access to matrix $G_i$, which can be chosen to be the square root matrix $\sqrt{\frac{V+V^{\dagger}}{2}}$ (here $D = 1$). 
Note that in this work we mainly consider the semi-dissipative case, which implies that $\frac{V+V^{\dagger}}{2}$ is positive semidefinite and $G$ can be uniquely defined.
We show how to construct the block encoding of $G$ as follows.

\begin{lemma}
    Assume we are given access to an $(\alpha,a,0)$-block encoding $U_B$ of a positive semi-definite matrix $B\in \mathbb{C}^{2^n\times 2^n}$.
    Suppose that the smallest non-zero eigenvalue of $G$ is lower bounded by $\Delta>0$.
    One can construct a $(2\sqrt{\alpha},a+2,\epsilon)$-block encoding of matrix $G$ such that $B=G^\dagger G$, by using $\mathcal{O}(\frac{\alpha}{\Delta}\log(\frac{\alpha}{\epsilon}))$ times of controlled-$U_B$ and controlled-$U_B^\dagger$.
\end{lemma}
\begin{proof}
    As $B$ is a positive semi-definite matrix, we have $B=\sum_i \lambda_i |\psi_i\rangle \langle \psi_i|$, where $\lambda_i\geq 0$ for $i\in [2^n]$ and $\{\ket{\psi_i}\}_{i=0}^{2^n-1}$ is the basis.
    Therefore, we have $G=\sum_i \sqrt{\lambda_i}|\psi_i\rangle \langle \psi_i|$, which is uniquely determined.
    This provides the intuition that if we can implement the square root function on $\lambda_i$'s, we can achieve the matrix $G$.

    By Lemma~\ref{poly.sqrt}, there is a real odd polynomial that can approximate the function $\frac{1}{2}\sqrt{x}$.
    Combining this with Lemma~\ref{theorem.qsvt}, one can construct a $(2\sqrt{\alpha},a+2,\epsilon)$-block encoding of $G$ by using $\mathcal{O}(\frac{\alpha}{\Delta}\log(\frac{\alpha}{\epsilon}))$ times of controlled-$U_G$ and controlled-$U_G^\dagger$, and with $\mathcal{O}(\mathrm{polylog}(\frac{\alpha}{\epsilon}))$ classical computation time.
\end{proof}

An additional point to notice is that for the main algorithm considered in this work, we do not only need to implement the square root function onto $G$ with the polynomial $P$, but we also require that for $x=0$, $P(0)=0$, i.e., the implemented transformation will not add an energy shift. 
If there is a energy shift, i.e., $P(0)\neq 0$, then $e^{-P(0)t}$ will decrease exponentially with $t$.
This can be achieved using the odd polynomial, which satisfies $P(0)=0$.

Before we proceed to the complexity analysis, we first show the following linear algebra result, which is useful in bounding the errors of our algorithm. 

\begin{lemma}[Norm propagation\label{lemma.norm.error}]
    Let $\vec{\psi}$, $\vec{\phi}$ and $\vec{\mu}$ be vectors. Then
    \begin{align}
        \|\vec{\psi}-\vec{\phi}\|_2\leq \|(\vec{\psi}-\vec{\phi})\vec{\mu}^{\dagger}\|_1/\|\vec{\mu}\|_2.
    \end{align}
\end{lemma}
\begin{proof}
We have 
    \begin{align}
    \|\vec{\psi}-\vec{\phi}\|_2&=\|((\vec{\psi}-\vec{\phi})\vec{\mu}^\dagger) \vec{\mu}\|_2/ \|\vec{\mu}\|_2^2\\
    &\leq \|(\vec{\psi}-\vec{\phi})\vec{\mu}^\dagger\|/ \|\vec{\mu}\|_2\\
    &\leq \|(\vec{\psi}-\vec{\phi})\vec{\mu}^\dagger\|_1/ \|\vec{\mu}\|_2,    
    \end{align}
    where the first inequality is from the definition of the spectral norm, and the last inequality is that spectral norm is smaller than the trace norm.
\end{proof}

Now we show how to solve the linear differential equation via the Lindbladian in the first input model: direct access model.

\begin{theorem}[Quantum ODE solver via Lindbladian, direct access model\label{alg.first}]
    Consider a time-independent, semi-dissipative, and homogeneous linear ordinary differential equation $\frac{d}{dt}\vec{\mu}(t)=-V\vec{\mu}(t)$ with $\vec{\mu}(0)=\ket{\mu_0}$, where $V\in \mathbb{C}^{2^n\times 2^n}$.
    Assume we are given access to $(\alpha_V,a_V,0)$-block encoding of $V$, and quantum state preparation unitary $U_{\mu_0}:\ket{0}_n\rightarrow |\mu_0\rangle$.
    Suppose for the Hermitian part of matrix $V$, the smallest non-zero eigenvalue is lower bounded by $\Delta>0$, then there exists a quantum algorithm that outputs an $\epsilon$-approximation of  $\ket{\mu_T}$ by using $\mathcal{O}\left(\Delta^{-1}\eta_T^{-1}\alpha_V^2 T\frac{\log^2(\alpha_V T/(\epsilon\eta_T))}{\log\log(\alpha_VT/(\epsilon\eta_T))}\right)$ times of controlled-$U_V$ and controlled-$U_V^\dagger$, and $\mathcal{O}(\eta_T^{-1})$ times of $U_{\mu_0}$, where $\eta_T$ is the normalization factor of $\vec{\mu}(T)$. Additionally, the algorithm also uses $\mathcal{O}\left(\eta_T^{-1} \left(\Delta^{-1}\alpha_V^2 T\text{poly}\log(\eta_T^{-1}\alpha_V T\epsilon^{-1})+\alpha_{V} T \text{poly}\log(\eta_T^{-1}\alpha_VT\epsilon^{-1})\right)\right)$ additional 1- and 2-qubit gates, and $\mathcal{O}(\text{poly}\log(\eta_T^{-1}\alpha_VT\epsilon^{-1}))$ number of additional ancilla qubits besides those used for block encodings.

\end{theorem}
\begin{proof}
    Given access to $U_V$, one can easily construct unitaries $U_A$ and $U_B$ which are $(\alpha_V,a_V+1,0)$-block encodings of $A=\frac{V-V^\dagger}{2i}$ and $B=\frac{V+V^\dagger}{2}$ respectively, by taking the linear combination using single time of controlled-$U_V$ and controlled-$U^\dagger_V$.

    Then we follow the step described in Section.~\ref{sec.encode.odetolind} to encode the linear equation into a Linbladian.
    We choose $\rho_0=\frac{1}{2}(\ket{0}|\mu_0\rangle+\ket{1}|\mu_0\rangle)(\bra{0}\langle \mu_0|+\bra{1}\langle \mu_0|)$, i.e., 
    \begin{align}
        \rho_0=\frac{1}{2}
        \begin{pmatrix}
            |\mu_0\rangle\langle \mu_0| & |\mu_0\rangle\langle \mu_0|\\
            |\mu_0\rangle\langle \mu_0| & |\mu_0\rangle\langle \mu_0|
        \end{pmatrix}
    \end{align}
    as the initial state.
    For the internal Hamiltonian $H$ and jump operator $F$, we set
    \begin{align}\label{eq.jumpop}
        H=\begin{pmatrix}
            A&0\\ 0&0
        \end{pmatrix}, 
        F=\begin{pmatrix}
            \sqrt{B} & 0\\ 0&0
        \end{pmatrix}.
    \end{align}
    Note that since we consider a semi-dissipative ODE, $B$ is semi-definite, and $\sqrt{B}$ is uniquely defined.
    To construct the block encodings of $H$ and $F$, notice that 
    \begin{align}
        H=(|0\rangle\langle 0|\otimes I)(I\otimes A),
    \end{align}
    and it is the same for $F$. 
    One can easily construct an $(1,1,0)$-block encoding of $|0\rangle\langle 0|$ by taking the linear combination of identity and Pauli $Z$ operator.
    By taking the product formula as Lemma.~\ref{product.blockencoding}, one can construct an $(\alpha_V,a_V+2,0)$-block encoding of $U_H$.
    It follows the same procedure for block encoding of $F$, except that we need to implement a square root function first.
    Based on the assumption, we can use Lemma~\ref{poly.sqrt} to construct an $(2\sqrt{\alpha_V}, a_V+5, \epsilon')$-block encoding of $G$ by using $\mathcal{O}(\frac{\alpha_V}{\Delta}\log(\frac{\alpha_V}{\epsilon'}))$ times of $U_B$, where $\epsilon'=\mathcal{O}(\epsilon/(\eta_T T))$.
    Then we can further construct a $(2\sqrt{\alpha_V}, a_V+6, \epsilon')$-block encoding $U_F$ of $F$. 
    
    Now, we implement the open system simulation based on Lemma~\ref{thm.opensystem.independent} to construct a purification of $\rho_T$, i.e, 
        \begin{align}
        \rho_T=\frac{1}{2} 
        \begin{pmatrix}
            \sigma_T & \eta_T |\mu_T\rangle\langle \mu_0|\\
            \eta_T |\mu_0\rangle\langle\mu_T| & |\mu_0\rangle\langle \mu_0|
        \end{pmatrix},
    \end{align}
    where $\sigma_T$ is the output density matrix under the evolution of Lindbladian from initial state $|\mu_0\rangle\langle \mu_0|$, and $\eta_T$ is the normalization factor of $|\mu_T\rangle$,
    by using $\mathcal{O}\left(\alpha_VT\frac{\log(\alpha_V\eta_T T/\epsilon)}{\log\log(\alpha_V\eta_TT/\epsilon)}\right)$ times of $U_H$ and $U_F$, where $\|\mathcal{L}\|_{\mathrm{be}}=\alpha_V+2\alpha_V=\mathcal{O}(\alpha_V)$.
    To achieve this, we use $\mathcal{O}\left(\frac{1}{\Delta}\alpha_VT\frac{\log(\alpha_V\eta_T T/\epsilon)}{\log\log(\alpha_V\eta_TT/\epsilon)}\right)$
    times of controlled-$U_H$ and controlled-$U_F^\dagger$, where $\epsilon$ is the error in the trace norm distance.
    
    Now, we output $\epsilon\eta_T$-approximation $\tilde{\rho}_T$ of purified state, i.e., $\|\tilde{\rho}_T-\rho_T\|_1\leq \epsilon\eta_T$.
    Let $\ket{S_T}$ be the purified quantum state of $\rho_T$, i.e., $\mathrm{Tr}_{\mathrm{e}}[\ket{S_T}\bra{S_T}]=\rho_T$, and let $\ket{\tilde{S}_T}$ be the state of our output, i.e., $\mathrm{Tr}_{\mathrm{an}}[\ket{\tilde{S}_T}\bra{\tilde{S}_T}]=\tilde{\rho}_T$.
    By definition, $\|\rho_T\|_1=\mathrm{Tr}[\ket{S_T}\bra{S_T}]=\|\ket{S_T}\|_2^2$ as density matrix is positive semi-definite.
    We have $\epsilon\eta_T\geq \|\rho_T-\tilde{\rho}_T\|_1=\|\mathrm{Tr}_{e}[\ket{S_T}\bra{S_T}-\ket{\tilde{S}_T}\bra{\tilde{S}_T}]\|_1$.
    If we only focus on the top-right block the density matrix, which corresponds to $|\mu_T\rangle\langle \mu_0|$, we have $(\bra{0}\otimes I_n)\rho_T(\ket{1}\otimes I_n)=\eta_T|\mu_T\rangle\langle \mu_0|$.
    We have 
    \begin{align}
        \|\eta_T|\mu_T\rangle\langle \mu_0|-\tilde{\eta}_T\ket{\tilde{\mu}(T)}\langle \mu_0|\|_1=&\|\mathrm{Tr}_{e}[(\bra{0}\otimes I_n)(\ket{S_T}\bra{S_T}-\ket{\tilde{S}_T}\bra{\tilde{S}_T})(\ket{1}\otimes I_n)]\|_1\\
        \leq&\|\mathrm{Tr}_{e}[|S_T\rangle \langle S_T|-\ket{\tilde{S}_T}\bra{\tilde{S}_T}]\|_1\leq \epsilon\eta_T.
    \end{align}
    Further, by Lemma~\ref{lemma.norm.error}, we have $\|\eta_T\ket{\mu_T}-\tilde{\eta}_T\ket{\tilde{\mu}_T}\|_2 \leq \|\eta_T|\mu_T\rangle\langle \mu_0|-\tilde{\eta}_T\ket{\tilde{\mu}(T)}\langle \mu_0|\|_1\leq \epsilon\eta_T$, since $\|\ket{\mu_0}\|_2=1$.
    Finally, by Lemma~\ref{lemma.state.extract} we prepare the state $\ket{\mu_T}$. In total, we use 
    \begin{align}
        \mathcal{O}\left(\frac{1}{\Delta}\frac{1}{\eta_T}\alpha_V^2 T\frac{\log^2(\alpha_V T/(\epsilon\eta_T))}{\log\log(\alpha_VT/(\epsilon\eta_T))}\right)
    \end{align} 
    times of controlled-$U_V$ and controlled-$U_V^\dagger$.

    Regarding the number of additional 1- and 2-qubit gates, the scaling is basically the number of amplitude amplification rounds ($\mathcal{O}(\eta_T^{-1})$) $\times$ (number of gates in QSVT for the square root transformation in Lindbladian simulation(i.e., the order of the query complexity of $U_V$) + number of gates in Lindbladian simulation algorithm (Lemma \ref{thm.opensystem.independent})). For the number of ancilla qubits beyond those for block encodings, it simply follow the results of the Lindbladian simulation (Lemma \ref{thm.opensystem.independent}).
\end{proof}


Now we discuss the second input model: square root access model.

\begin{theorem}[Quantum ODE solver via Lindbladian, square root access model\label{alg.second}]
    Consider a time-independent, semi-dissipative, and homogeneous linear ordinary differential equation $\frac{d}{dt}\vec{\mu}(t)=-V\vec{\mu}(t)$ with $u(0)=u_0$, where $V\in \mathbb{C}^{2^n\times 2^n}$.
    Assume we are given a $1$-block encoding $U_{H_1,G_i}$ of $\alpha_V^{-1}|0\rangle\langle 0|\otimes H_1+\sum_i \alpha_V^{-1/2}|i\rangle\langle i|\otimes G_i$ such that 
    \begin{align}
        (\bra{0}_{a+b}\otimes I)U_{H,G}(\ket{0}_{a+b}\otimes I)&=\frac{H_1}{\alpha_V},\\
        (\bra{0}_{a}\bra{i}_b\otimes I)U_{H,G}(\ket{0}_{a}\ket{i}_b\otimes I)&=\frac{G_i}{\sqrt{\alpha_V}},
    \end{align}
    where $H_1=\frac{V-V^\dagger}{2i}$ and $\sum_{i=1}^D G_iG_i^\dagger =\frac{V+V^\dagger}{2}$. We also assume access to a quantum state preparation unitary $U_{u_0}:\ket{0}\rightarrow |\mu_0\rangle$.
    Then, there exists a quantum algorithm that outputs an $\epsilon$-approximation of $\ket{\mu_T}$ by using $\mathcal{O}\left(\eta_T^{-1}D\alpha_V T\frac{\log(D\alpha_V T/(\epsilon\eta_T))}{\log\log(D\alpha_VT/(\epsilon\eta_T))}\right)$ times of $U_{H_1,G_i}$, and $\mathcal{O}(\eta_T^{-1})$ times of $U_{\mu_0}$, where $\eta_T$ is the normalization factor of $\vec{\mu}(T)$. Additionally, the algorithm also uses $\mathcal{O}\left(\eta_T^{-1} D^2\alpha_{V} T \text{poly}\log(\eta_T^{-1}D\alpha_VT\epsilon^{-1})\right)$ additional 1- and 2-qubit gates, and $\mathcal{O}(\text{poly}\log(\eta_T^{-1}D\alpha_{V}T\epsilon^{-1}))$ number of additional ancilla qubits besides those used for block encodings.
\end{theorem}

\begin{proof}
    The algorithm follows similarly to Theorem~\ref{alg.first}, but is greatly simplified due to the input model.
    One can construct $U_H$ and $U_F$ with the same step in Theorem~\ref{alg.first}, which are $(\alpha_V,a_V+1,0)$-block encodings of $H$ and $(\sqrt{\alpha_V},a_V+1,0)$-block encodings of $F$ respectively, where $H$ and $F$ is as Eq.~\ref{eq.jumpop}.
    We set the initial state $\rho_0=\frac{1}{2}(\ket{0}|\mu_0\rangle+\ket{1}|\mu_0\rangle)(\bra{0}\langle \mu_0|+\bra{1}\langle \mu_0|)$. 
    Then we use Theorem~\ref{thm.opensystem.independent} to construct a purification of $\rho_T$, where $\|\mathcal{L}\|_{\mathrm{be}}=\alpha_V+\frac{1}{2}\sum_{i=1}^D \alpha_V=\mathcal{O}(D\alpha_V)$. Finally, we use Lemma~\ref{lemma.state.extract} to extract the state $\ket{\mu_T}$. A combination of all these procedures gives the complexity in the theorem. Also, since we get rid of the use of QSVT, the number of additional gates will simply be the number of the amplitude amplification rounds $\times$ the number of gates in Lindbladian simulation (Lemma \ref{thm.opensystem.independent})
\end{proof}

\subsubsection{Time-dependent case}

We continue to show for the time-dependent and homogeneous case, with benefit from the previous results on time-dependent Lindbladian simulation \cite{he2024efficientoptimalcontrolopen}.
We first describe the two input models we consider in this work.

\textit{Time-dependent direct access}.
For $0\leq t\leq T$, where $T$ is the target evolution time, let $\alpha_V\geq \max_t \|V(t)\|$.
Assume we are given access to $(\alpha_V,a,0)$-block encoding $U_{V(t)}$ of matrix $V(t)$, i.e.,
\begin{align}
    (\bra{0}_a\otimes I)U_{V(t)}(\ket{0}_a\otimes I)=\sum_{t} |t\rangle\langle t|\otimes \frac{V(t)}{\alpha_V}.
\end{align}
We also assume $\alpha_V$ to be known classically, which enables us to construct the following oracles
\begin{align}
O_{\alpha_{V}}\ket{z}&=\ket{z\oplus \alpha_V^{-1}},\\
O_{H,\mathrm{norm}}\ket{z}&=\ket{z\oplus \alpha_V},\\
O_{G,\mathrm{norm}}\ket{z}&=\ket{z\oplus \sqrt{\alpha_V}}.
\end{align}

\textit{Time-dependent square root access}.
For $0\leq t\leq T$, where $T$ is the target evolution time, 
assume we are given access to $1$-block encoding $U_{H_1(t),G_i(t)}$ of $\alpha_V^{-1}|0\rangle\langle 0|\otimes H_1(t)+\sum_i \alpha_V^{-1/2}|i\rangle\langle i|\otimes G_i(t)$, i.e.,
    \begin{align}
        (\bra{0}_{a+b}\otimes I)U_{H_1(t),G_i(t)}(\ket{0}_{a+b}\otimes I)&=\sum_{t} |t\rangle\langle t|\otimes\frac{H_1(t)}{\alpha_V},\\
        (\bra{0}_{a}\bra{i}_b\otimes I)U_{H_1(t),G_i(t)}(\ket{0}_{a}\ket{i}_b\otimes I)&=\sum_{t} |t\rangle\langle t|\otimes\frac{G_i(t)}{\sqrt{\alpha_V}},
    \end{align}
    where $H_1(t)=\frac{V(t)-V^\dagger(t)}{2i}$ and $\sum_{i=1}^D G_i^\dagger(t)G_i(t) =\frac{V(t)+V^\dagger(t)}{2}$.

\begin{corollary}[Quantum algorithm for the time-dependent ODE, direct access model]
    Consider a time-dependent, semi-dissipative, and homogeneous linear ordinary differential equation $\frac{d}{dt}\vec{\mu}(t)=-V(t)\vec{\mu}(t)$, where $V\in \mathbb{C}^{2^n\times 2^n}$.
    Assume we are given access to the time-dependent direct access and quantum state preparation unitary $U_{\mu_0}:\ket{0}\rightarrow \ket{\mu_0}$.
    Suppose for the Hermitian part of the matrix $V(t)$, the smallest non-zero eigenvalue is lower bounded by $\Delta>0$.
    Then there exists a quantum algorithm that outputs an $\epsilon$-approximation of $\ket{\mu_T}$ by using $\mathcal{O}\left(\Delta^{-1}\eta_T^{-1}\alpha_V^2 T\frac{\log^3(\alpha_V T/(\epsilon\eta_T))}{\log^2\log(\alpha_VT/(\epsilon\eta_T))}\right)$ times of $U_{V(t)}$, and $\mathcal{O}(\eta_T^{-1})$ times of $U_{\mu_0}$, where $\alpha_V\geq \max_{t} \|V(t)\|$ and $\eta_T$ is the normalization factor of $\vec{\mu}(T)$. Additionally, the algorithm also uses $\mathcal{O}\left(\eta_T^{-1} \left(\Delta^{-1}\alpha_V^2 T\text{poly}\log(\eta_T^{-1}\alpha_V T\epsilon^{-1})+n\alpha_{V} T \text{poly}\log(\eta_T^{-1}(\alpha_V+\beta_V)T\epsilon^{-1})\right)\right)$ additional 1- and 2-qubit gates, and $\mathcal{O}(\text{poly}\log(\eta_T^{-1}(\alpha_V+\beta_V)T\epsilon^{-1}))$ number of additional ancilla qubits besides those used for block encodings, where we define $\beta_V=\max_{t \in [0,T]} \|\frac{\partial( V(t)-V(t)^\dagger)}{\partial t}\| +  \max_{t\in[0,T]} \|\frac{\partial \sqrt{V(t)+V(t)^\dagger}}{\partial t}\|$.
\end{corollary}
\begin{proof}
    Following the same steps as Theorem~\ref{alg.first} and adapting the Linbladian simulation part with Lemma~\ref{thm.opensystem.dependent.simple}, one can achieve the claim.
    This is true because Lemma~\ref{thm.opensystem.dependent.simple} as shown in Ref~\cite{he2024efficientoptimalcontrolopen} achieves the time-dependent open system simulation based on truncated Dyson series, and for each discretized $t$, we can construct the needed block encodings. Also note that the definition of $\beta_V$ simply follows the definition of $\beta_{\mathcal{L}}$ in Lemma \ref{thm.opensystem.dependent.simple}.
\end{proof}

\begin{corollary}[Quantum algorithm for the time-dependent ODE, square root access model]
    Consider a time-dependent, semi-dissipative, and homogeneous linear ordinary differential equation $\frac{d}{dt}\vec{\mu}(t)=-V(t)\vec{\mu}(t)$, where $V\in \mathbb{C}^{2^n\times 2^n}$.
    Assume we are given access to the time-dependent square root access and quantum state preparation unitary $U_{\mu_0}:\ket{0}\rightarrow \ket{\mu_0}$.
    Then there exists a quantum algorithm that outputs an $\epsilon$-approximation of $\ket{\mu_T}$ by using $\mathcal{O}\left(\eta_T^{-1}D\alpha_V T\frac{\log^2(D\alpha_V T/(\epsilon\eta_T))}{\log^2\log(D\alpha_VT/(\epsilon\eta_T))}\right)$ times of $U_{H_1(t),G_i(t)}$, and $\mathcal{O}(\eta_T^{-1})$ times of $U_{\mu_0}$, where $\alpha_V\geq \max_t \|V(t)\|$ and $\eta_T$ is the normalization factor of $\vec{\mu}(T)$. Additionally, the algorithm also uses $\mathcal{O}\left(\eta_T^{-1} (n+D)D\alpha_{V} T \text{poly}\log(\eta_T^{-1}(D\alpha_V+\beta_V')T\epsilon^{-1})\right)$ additional 1- and 2-qubit gates, and $\mathcal{O}(\text{poly}\log(\eta_T^{-1}(D\alpha_V+\beta_V')T\epsilon^{-1}))$ number of additional ancilla qubits besides those used for block encodings, where we define $\beta_V'=\max_{t \in [0,T]} \|\frac{\partial H_1(t)}{\partial t}\| + \sum_{i=1}^D \max_{t\in[0,T]} \|\frac{\partial G_i(t)}{\partial t}\|$
\end{corollary}
\begin{proof}
    Following the same steps as Theorem.~\ref{alg.second} and adapting the Linbladian simulation part with Lemma~\ref{thm.opensystem.dependent.simple}, one can achieve the claim.
\end{proof}

\section{Lower bound discussion\label{sec:lowb}}
In this section, we discuss the complexity lower bound of quantum DOE solvers with respect to $\eta_T$, $\alpha_V T$, and $\epsilon$ under different input models. The basic result is:
\begin{theorem}[Complexity lower bound]
The complexity lower bound under the square root access model is the same as the direct access model.
\end{theorem}
\begin{proof}
For $\eta_T$, both the queries to $V(t)$ and the queries to $|\mu_0\rangle$ (initial state preparation oracle) have a lower bound dependence $\mathcal{O}(\eta_T^{-1})$ under the first input model. For $|\mu_0\rangle$, this lower bound is related to using ODE solver for the state discrimination
task \cite{an2023theoryquantumdifferentialequation}. For $V(t)$, this lower bound is related to the optimality of the Grover speedup for unstructured database searching \cite{fang2023time,bennett1997strengths}. Both reasons come from the linearity of quantum mechanics \cite{abrams1998nonlinear,childs2016optimal}. Under the second input model, the same lower bound holds for both $V(t)$ and $|\mu_0\rangle$. For $|\mu_0\rangle$, since the lower bound comes from the distinguishability between two non-orthogonal quantum states, it has nothing to do with the input model of $V(t)$, thus the lower bound is still $\mathcal{O}(\eta_T^{-1})$. For $V(t)$, given an unstructured searching oracle:
$$U_{sear}|x_0\rangle=|x_0\rangle\text{, }U_{sear}|x\rangle=-|x\rangle\text{ for $x\neq x_0$},$$
with $U_{sear}$ unitary and Hermitian, we can use LCU to have a block encoding of $(U_{sear}+I)/2$ which has eigenvalues of either 0 or 1. Under the first input model, we can let $V=(U_{sear}+I)/2$, and run the algorithm for the searching task, which will give the $\mathcal{O}(\eta_T^{-1})$ lower bound as shown in Ref. \cite{fang2023time}. Under the second input model, since $\sqrt{(U_{sear}+I)/2}=(U_{sear}+I)/2$, this has no difference from the first, thus, the same lower bound holds.

Under the first input model, both $\alpha_V T$ and $\epsilon$ have lower bounds with respect to queries of $V(t)$: $\mathcal{O}(\alpha_V T)$ and $\mathcal{O}(\frac{\log(\epsilon^{-1})}{\log\log(\epsilon^{-1})})$. Both lower bounds of $\alpha_V T$ and $\epsilon$ come from the no fast-forwarding theorem \cite{berry2007efficient,berry2014exponential} related to the complexity of determining the
parity of bits \cite{beals2001quantum}. These lower bounds come from the Hamiltonian simulation of $e^{-iHt}$ with $H$ given by the sparse Hamiltonian access oracle \cite{berry2014exponential, low2019hamiltonian}. Unlike the Hermitian part where we need the square root function, the query of $H_1(t)=-i(V(t)-V(t)^\dag)/2$ is equivalent to the query on $V(t)$ with the aid of LCU. Thus, the lower bounds come from the sparse Hamiltonian simulation naturally hold for the second input model.

Since the oracle $U_{H_1(t),G_i(t)}$ is a combination of $H_1(t)$ and $\{G_i(t)\}$, the lower bound of $\eta_T$ from $\{G_i(t)\}$ and the lower bounds of $\alpha_VT$ and $\epsilon$ from $H_1(t)$ are inherited by $U_{H_1(t),G_i(t)}$, and thus it has the same lower bounds for all parameters as the first input model.
\end{proof}
\section{Relations between Lindbladians and non-Hermitian physics\label{sec:rela}}
The construction in our work can be considered as a bridge between non-Hermitian physics and open quantum systems, which further implies the importance of our work. In recent years, there have been many studies on exploring non-Hermitian effects in open quantum systems \cite{song2019non,roccati2022non,minganti2019quantum}. Typically, in these works, for the Lindbladian
$$\frac{d\rho}{dt}=\mathcal{L}[\rho]=-i[H,\rho]+
\sum_i \left(G_i\rho G_i^\dag-
\frac{1}{2}\{\rho,G_i^\dag G_i\}\right),$$
the authors will consider an effective non-Hermitian Hamiltonian $H_{eff}$ with the form
$$H_{eff}=H-\frac{i}{2}\sum_i G_i^\dag G_i,$$ 
to characterize the dynamics of Lindbladian. Various phenomena around this non-Hermitian $H_{eff}$ have been investigated. However, this effective Hamiltonian can only have a good approximation to the exact dynamics in the semi-classical limit \cite{monkman2024limits} or for short times (i.e., before the occurrence of a jump). We can understand this from the following stochastic differential equation (SDE) \cite{jacobs2006straightforward,ding2024simulating}
\begin{equation}\label{ssdde}
d|\psi_t\rangle=-i H_{eff}|\psi_t\rangle+\sum_i G_i |\psi_t\rangle dW_{i,t},
\end{equation}
with $dW_{i,t}$ independent Wiener processes. One can prove that the average dynamics of this SDE is
\begin{equation}
\frac{d\mathbb{E}[|\psi_t\rangle\langle \psi_t|]}{dt}-i[H,\mathbb{E}[|\psi_t\rangle\langle \psi_t|]]+
\sum_i \left(G_i\mathbb{E}[|\psi_t\rangle\langle \psi_t|] G_i^\dag-
\frac{1}{2}\{\mathbb{E}[|\psi_t\rangle\langle \psi_t|],G_i^\dag G_i\}\right),
\end{equation}
which is exactly the form of the Lindbladian. Only at a short time (before the occurrence of the Wiener jumps in Eq. \ref{ssdde}), the SDE can be reduced to the Schrodinger equation with $H_{eff}$ as the non-Hermtian Hamiltonian. In contrast, in our work, for the first time, we show that by carefully designing Lindbladians, we can recover the exact dynamics of $H_{eff}$ in a non-diagonal block of a density matrix, indicating a new relation between effective Hamiltonian and Lindbladian, which combining with our algorithm implementation, makes the quantum computer a promising test bed to explore non-Hermitian physics.

\end{appendix}
\end{document}